\documentclass[12pt]{article}
\usepackage{amsmath}
\usepackage{mathtools}
\usepackage{amsthm}
\usepackage{amssymb}
\usepackage{authblk}
\usepackage{bm}
\usepackage{hyperref}
\usepackage{booktabs}
\usepackage{color}
\usepackage{comment}
\usepackage{graphicx,psfrag,epsf}
\usepackage{enumerate}
\usepackage{epstopdf}
\usepackage[authoryear,round]{natbib}
\usepackage{tabularx}
\usepackage{ulem}
\usepackage{url} % not crucial - just used below for the URL
\usepackage{xcolor}

\newcommand{\Sigmar}{\bm{\Sigma}_{\text{r}}}
\newcommand{\Sigmac}{\bm{\Sigma}_{\text{c}}}
\newcommand{\blind}{0}
\newcount\Comments  % 0 suppresses notes to selves in text
\definecolor{darkgreen}{rgb}{0,0.5,0}
\definecolor{purple}{rgb}{1,0,1}
\newcommand{\kibitz}[2]{\ifnum\Comments=1\textcolor{#1}{#2}\fi}

\newcommand{\tv}[1]{\textsc{vol}_{#1}}
\newcommand{\bp}[1]{\textsc{rbp}_{#1}}
\newcommand{\ret}[1]{\textsc{ret}_{#1}}
\newcommand{\e}[3]{#1_{#3}^{\textsc{#2}}}
\newcommand{\ee}[2]{\tilde{e}_{#2}^{\textsc{#1}}}

\newcommand{\note}[1]{\begin{flushleft}\footnotesize\textbf{Note:} #1\end{flushleft}}

\newtheorem{proposition}{Proposition}

\author[1]{Andrea Bucci}
\author[2]{Giulio Palomba}
\author[3]{Eduardo Rossi}

\date{}

\affil[1]{\footnotesize Department of Economics and Law, University of Macerata, Italy}
\affil[2]{\footnotesize Department of Economics and Social Sciences, Marche Polytechnic University, Italy}
\affil[3]{\footnotesize Department of Economics and Management, University of Pavia, Italy}

\begin{document}

\def\spacingset#1{\renewcommand{\baselinestretch}%
{#1}\small\normalsize} \spacingset{1}

%%%%%%%%%%%%%%%%%%%%%%%%%%%%%%%%%%%%%%%%%%%%%%%%%%%%%%%%%%%%%%%%%%%%%%%%%%%%%%

\if0\blind
{
  \title{\bf A Structural Matrix Autoregressive Model for the Joint Dynamics of Volume, Volatility, and Returns
  }
  \maketitle
} \fi

\if1\blind
{
  \bigskip
  \bigskip
  \bigskip
  \begin{center}
    {\LARGE\bf Title}
\end{center}
  \medskip
} \fi

\bigskip
\begin{abstract}
  This paper proposes a Structural Matrix Autoregressive (SMAR) model for the joint analysis of asset returns, realized volatility, and trading
  volume in a large-dimensional setting. This framework simultaneously captures dynamic spillovers across financial variables and
  cross-sectional dependence across assets while preserving a parsimonious parameterization relative to conventional vector autoregressive
  models. The model is estimated on daily data for the constituents of the Dow Jones Industrial Average over the period 2021--2025 and is
  structurally identified through restrictions consistent with the Mixture of Distributions Hypothesis and efficient market theory. The empirical
  findings indicate that volatility is the primary driver of trading activity, suggesting that informational shocks are predominantly
  incorporated into markets through price variability. Forecast error variance decompositions further reveal that, although internal shocks
  dominate short-term volume dynamics, cross-asset spillovers account for more than 50\% of trading volume variation at longer horizons. Finally,
  an event-study analysis around FOMC announcements supports the proposed decomposition by identifying significant increases in the informative
  component of trading activity on announcement days followed by rapid mean reversion.
\end{abstract}

\noindent%

\newpage
\spacingset{1.8}

\section{Introduction}

Recent global history has shown how extreme events cause financial markets to jitter, which experience periods of high volatility and extremely
variable trading volumes \citep[see][]{Longin2001, Diebold2009, Bollerslev2018}. In this context, it is crucial for market participants,
institutions, and policymakers to understand how volatility and liquidity spillovers interact \citep{Acharya2005, Brunnermeier2009}. Indeed,
while from a microeconomic perspective, understanding the dynamics and correlations of asset prices leads investors and asset managers to
portfolio optimization, asset management, and risk hedging, from a macroeconomic perspective, it should lead policymakers towards strategies
aimed at economic and financial stability.

Although traditional financial theories have primarily focused on the dynamics between asset prices and their volatility, the analysis of trading
volumes has received relatively less attention, despite being an indicator of market sentiment \citep{Darolles2015}. For example, price increases
during periods of low volume may only produce transitory effects, while similar price movements during periods of high volume may have
long-lasting implications. As further discussed in \cite{Bollerslev2018}, trading volume and return volatility tend to co-vary
\citep{Chiang2010}, especially after news shocks. This result underlies the common practice of using trading volume as an exogenous regressor in
time varying volatility models.

Furthermore, the mixed distribution hypothesis (MDH), initially formulated by \cite{Tauchen1983}, posits that the observed joint behavior between
asset volatility and trading volumes is due to latent information flows that simultaneously drive price fluctuations and trading activity. In
parallel, \cite{Darolles2015} emphasize the methodological challenges in jointly modeling price dynamics and trading volumes, particularly when
one wants to separate the dynamics related to information shocks from those attributable to liquidity constraints. These insights point to the
necessity of extending existing approaches to better capture the complex dynamics underlying trading activity. In this direction, \cite{He2014}
develop a multi-asset mixture distribution hypothesis model to investigate commonality in stock returns and trading volumes based on factor
structure for returns and trading volumes.

Although existing methodologies reliably assess the liquidity-volatility relationship for individual stocks \citep{Chordia2000, Pastor2003}, they
overlook different components of liquidity dynamics. To fill this gap in the literature, we propose here to model the cross-asset dynamics of
volatility, volumes and prices through the Structural Matrix-valued AutoRegressive (SMAR) model. This method extends the structural vector
autoregressive (SVAR) model of \cite{Sims1980} to multivariate settings in which data form matrices. As in the MAR model introduced by
\cite{Chen2021}, we collect $m$ assets in the rows and $n$ financial metrics in the columns, thus forming a data matrix for each time $t=1,2,\ldots,T$. In this way, SMAR jointly takes into account the spillover effects across assets and the dynamic interconnectedness of some
financial variables that share a common latent structure in the MDH. The SMAR offers a parsimonious alternative to the traditional architectures,
typically the Structural VAR (SVAR) or Panel VAR \citep[PVAR, see][for a comprehensive survey]{Canova2013} models, in which matrix data are
stacked into a vector. In fact, excluding any deterministic component, the vectorization in these models foresees the estimation of $m^2n^2$
parameters per lag, while the number of parameters in the SMAR reduces to $m^2+n^2$. Hence, the vectorization results in a loss of efficiency since the
relationships between rows and columns can no longer be maintained, and makes the estimation computationally cumbersome. In addition, while
identifying structural innovations in panel SVAR models typically requires computationally intensive data-pooling or the assumption of
statistical independence on a massive $mn \times mn$ covariance matrix $\bm{\Sigma}$ \citep[e.g.,][]{Herwartz2024}, the SMAR framework leverages a
separable covariance structure $\bm{\Sigma} = \Sigmar \otimes \Sigmac$, where $\Sigmac$ and $\Sigmar$ are the column- and row-wise error
covariance matrices. This allows one to perform simultaneous structural identification exclusively within the lower-dimensional $n \times n$
variable space $\Sigmac$, leaving cross-sectional dependencies to be absorbed by $\Sigmar$ without requiring restrictive entity-level
orderings. Furthermore, we can exploit the algebraic properties of our framework to solve analytically the identification issues and compute both
single-asset-single-variable impulse response functions (IRFs) and averaged IRFs across either assets or financial variables. In addition, since
our approach relies on the explicit structural impulse response functions, they represent a valid alternative to those defined by \cite{Chen2021}
in the MAR context.

We apply the SMAR model to daily data for $m=30$ assets composing the Dow Jones index, covering the period from January 2021 to February
2025. The empirical analysis delivers several findings. First, volatility is a key determinant of trading activity. Second, a forecast error
variance decomposition suggests that trading volume is primarily driven by its own shocks over the short run, while cross-asset spillovers
account for more than half of the total volume variance over a 20-day time horizon. Moreover, an event study around the announcements of the
Federal Open Market Committee (FOMC) shows a statistically significant spike in the information share of volume on announcement days
\citep{Lucca2015}, followed by a rapid mean reversion consistent with efficient price discovery \citep{Andersen2003}.

The remainder of the paper is organized as follows. In Section \ref{sec:Model}, we introduce the SMAR model, discuss identification issues, and
define the impulse response functions. Section \ref{sec:SMARfinance} proposes a specification of the column-wise structural matrix based on some
results from the finance literature and provides some financial considerations. Section \ref{sec:Empirical_application} is devoted to the
empirical analysis, and finally Section \ref{sec:Conclusions} concludes.

\section{The Structural MAR model}
\label{sec:Model}

Let $\mathbf{Y}_t$ be a time series matrix $m \times n$, with $t=1,2,\ldots,T$, containing data for $m$ assets and $n$ financial measures. Since
we are interested in a structural model that incorporates simultaneous relationships between all row and column variables of $\mathbf{Y}_t$,
these simultaneous interactions are represented in a linear setup by the $m \times m$ matrix $\mathbf{R}_0$ and the $(n \times n)$ matrix
$\mathbf{C}_0$, both of full rank. Therefore, the structural MAR($p$) model (hereafter SMAR) is defined as
\begin{equation}
  \label{eq:MARR0C0}
  \mathbf{R}_0\mathbf{Y}_t\mathbf{C}_0^\top = \mathbf{R}_1 \mathbf{Y}_{t-1} \mathbf{C}_1^\top + \dots + \mathbf{R}_p \mathbf{Y}_{t-p}\mathbf{C}_p^\top + \mathbf{U}_t,
\end{equation}
where $p$ is the lag order, $\mathbf{R}_1, \ldots, \mathbf{R}_p$ are $m \times m$ row-wise effect parameter matrices whose Frobenius norm must be
equal to 1 \citep[see][]{Chen2021}, $\mathbf{C}_1, \ldots, \mathbf{C}_p$ are $n \times n$ matrices containing the parameters for column-wise
relationships. The deterministic components have been suppressed for notational convenience. \(\mathbf{U}_t\) denotes an $m \times n$ matrix of
serially uncorrelated structural shocks with zero mean and unit variances \citep{kilian2017}. It follows that the $mn \times mn$ covariance
matrix of $\text{vec}(\mathbf{U}_t)$ (where the vec operator transforms a matrix into a vector by stacking the columns of the matrix one
underneath the other) is the identity matrix
\begin{equation}
  \label{eq:SigmaD}
  \bm{\Omega} = \text{Var}\left(\text{vec}(\mathbf{U}_t)\right) = \mathbf{I}_{mn} = \mathbf{I}_n \otimes \mathbf{I}_m,
\end{equation}
where $\mathbf{I}_{m}$ and $\mathbf{I}_{n}$ represent the row-wise and the column-wise covariance matrices, respectively. Therefore, there are
as many structural shocks as variables in the model, and these shocks are mutually uncorrelated with unit variance \citep{kilian2017}.

The reduced-form representation of the structural MAR($p$) model, \textit{i.e.}, where $\mathbf{Y}_t$ is a function of the lags of $\mathbf{Y}_t$
only, is obtained by post-multiplying both sides of Equation \eqref{eq:MARR0C0} by \(\mathbf{C}_0^{-\top}\) and pre-multiplying by
\(\mathbf{R}_0^{-1}\). This leads to
\[
  \mathbf{R}_0^{-1}\mathbf{R}_0\mathbf{Y}_t\mathbf{C}_0^\top\mathbf{C}_0^{-\top} = 
  \mathbf{R}_0^{-1}\mathbf{R}_1 \mathbf{Y}_{t-1} \mathbf{C}_1^\top \mathbf{C}_0^{-\top} + \dots +
  \mathbf{R}_0^{-1}\mathbf{R}_p \mathbf{Y}_{t-p}\mathbf{C}_p^\top \mathbf{C}_0^{-\top} + \mathbf{R}_0^{-1}\mathbf{U}_t\mathbf{C}_0^{-\top},
\] 
that corresponds to the reduced-form MAR($p$) by \cite{Chen2021} given by
\begin{equation}
    \label{eq:reduced}
        \mathbf{Y}_t = \mathbf{A}_1 \mathbf{Y}_{t-1} \mathbf{B}_1^\top + \dots + \mathbf{A}_p \mathbf{Y}_{t-p} \mathbf{B}_p^\top + \mathbf{E}_t,
\end{equation}
where $\mathbf{A}_i=\mathbf{R}_0^{-1}\mathbf{R}_i$ are $m\times m$ matrices, and $\mathbf{B}_i=\mathbf{C}_0^{-1}\mathbf{C}_i$ are $n\times n$
matrices for $\forall\,i = 1,\ldots,p$. Consequently, the reduced-form innovation matrix at time $t$ is given by
\begin{equation}\label{eq:Et}
  \mathbf{E}_t = \mathbf{R}_0^{-1}\mathbf{U}_t\mathbf{C}_0^{-\top}.
\end{equation}
To compute the covariance of $\mathbf{E}_t$, we can use the vectorized form of Equation \eqref{eq:MARR0C0}, given by
\begin{equation}\label{eq:vec1}
    \begin{split}
      (\mathbf{C}_0 \otimes \mathbf{R}_{0})\text{vec}(\mathbf{Y}_t) & = (\mathbf{C}_1 \otimes \mathbf{R}_1)\text{vec}(\mathbf{Y}_{t-1}) + \dots +
      (\mathbf{C}_p \otimes \mathbf{R}_p)\text{vec}(\mathbf{Y}_{t-p}) + \text{vec}(\mathbf{U}_t).
    \end{split}
\end{equation}
Since $\mathbf{R}_0$ and $\mathbf{C}_0$ are nonsingular, pre-multiplying by $\left(\mathbf{C}_0\otimes \mathbf{R}_0\right)^{-1}$ implies that
\begin{equation}\label{eq:vec2}
    \begin{split}
      \text{vec}(\mathbf{Y}_t) & = (\mathbf{C}_0 \otimes \mathbf{R}_{0})^{-1}(\mathbf{C}_1 \otimes \mathbf{R}_1)\text{vec}(\mathbf{Y}_{t-1}) +\\
      &+ \dots + (\mathbf{C}_0 \otimes \mathbf{R}_{0})^{-1}(\mathbf{C}_p \otimes \mathbf{R}_p)\text{vec}(\mathbf{Y}_{t-p}) + (\mathbf{C}_0
      \otimes \mathbf{R}_{0})^{-1}\text{vec}(\mathbf{U}_t).
    \end{split}
\end{equation}
Further, we can obtain
\[
  \text{vec}(\mathbf{E}_t) = \text{vec}\left(\mathbf{R}_0^{-1}\mathbf{U}_t \mathbf{C}_0^{-\top}\right) =
  \left[\mathbf{C}_0^{-1} \otimes \mathbf{R}_0^{-1}\right]\text{vec}(\mathbf{U}_t).
\]
Considering the covariance matrix of structural shocks defined by the equation \eqref{eq:SigmaD}, the covariance matrix of
\(\text{vec}(\mathbf{E}_t)\) turns out to be
\begin{equation}
  \label{eq:sigmaE}
  \bm{\Sigma}_E = \left[\mathbf{C}_0^{-1} \otimes \mathbf{R}_0^{-1}\right] \bm{\Omega} \left[\mathbf{C}_0^{-1} \otimes \mathbf{R}_0^{-1}\right]^\top=
  \left(\mathbf{C}_0^\top\mathbf{C}_0\right)^{-1}\otimes\left(\mathbf{R}_0^\top\mathbf{R}_0\right)^{-1}.
\end{equation}
Assuming that the covariance of $\mathbf{E}_t$ in a MAR model \citep{Chen2021} is an $mn \times mn$ matrix expressed as
\[
  \bm{\Sigma}_E = \Sigmac \otimes \Sigmar,
\]
it follows that the term $\Sigmac = \left(\mathbf{C}_0^\top\mathbf{C}_0\right)^{-1}$ represents the column-wise covariance of $\mathbf{E}_t$,
while $\Sigmar = \left(\mathbf{R}_0^\top\mathbf{R}_0\right)^{-1}$ denotes the row-wise covariance of $\mathbf{E}_t$. In this context, the
column-wise and row-wise covariance matrices of the reduced-form error, $\Sigmac$ and $\Sigmar$, can be considered known \citep{kilian2017}.

\subsection{Identification}
\label{sec:identify}

For the identification of $\mathbf{C}_0^{-\top}$, it can be noticed that the number of free parameters in $\Sigmac$ is $n(n +1)/2$ which is the
upper limit for the number of parameters in $\mathbf{C}_0^{-\top}$ that can be uniquely identified \citep{kilian2017}. This means that we have to
impose at least $n(n-1)/2$ normalizing restrictions.

Since $\Sigmac^{-1} = \mathbf{C}_0^\top \mathbf{C}_0$, we can assume that $\mathbf{C}_0^\top$ is a lower triangular matrix with positive diagonal
elements, as in the following section, \textit{i.e.}, the Cholesky factor of $\Sigmac^{-1}$.

Analogously, given that $\Sigmar$ has $m(m+1)/2$ unique elements, we must impose $m(m-1)/2$ conditions to achieve identification of the
parameters of $\mathbf{R}_0$. This means that, in our model, identifying contemporaneous effects across assets can be challenging, as the number
of required constraints increases quadratically with the number of assets ($m$). Defining such a large number of constraints from economic theory
is often very difficult, as the simultaneous interactions between different assets could be extremely complex.

While we recognize that constraints can be imposed in the Equation \eqref{eq:MARR0C0} by dividing assets into clusters using appropriate
techniques (\textit{e.g.}, cluster analysis or principal components analysis), or by adopting some criterion to separate the elements
(\textit{e.g.}, sector or return on assets), we assume that contemporaneous changes are due exclusively to structural shocks to financial
measures, rather than being specific to individual assets; in practice, we impose the structure $\mathbf{R}_0 = \mathbf{I}_m$.

\subsection{SMAR model estimation}
\label{sec:Estimation}

In this section, we present the log-likelihood of the SMAR(1) model. We start by defining the set of parameters contained in the vector
$\bm{\psi} = \left[\text{vec}(\mathbf{C}_0)^\top, \text{vec}(\mathbf{C}_1)^\top, \text{vec}(\mathbf{R}_0)^\top,
  \text{vec}(\mathbf{R}_1)^\top\right]^\top$. The Gaussian log-likelihood \citep{Wu2026} is
\begin{equation}\label{eq:LL}
\begin{split}
  \ell(\bm\psi; \mathbf{Y}_t, \mathbf{Y}_{t-1}) &= -\frac{m}{2}(T-1)\log |\Sigmac| - \frac{n}{2}(T-1)\log |\Sigmar| -
  \frac{1}{2}\sum_{t=2}^T\text{tr}\left(\Sigmar^{-1}\mathbf{E}_t \Sigmac^{-1}\mathbf{E}_t^{\top}\right),
\end{split}
\end{equation}
where
\begin{equation}
  \label{eq:ourSMAR}
  \mathbf{E}_t=\mathbf{Y}_t-(\mathbf{R}_0^{-1}\mathbf{R}_1)\mathbf{Y}_{t-1}(\mathbf{C}_0^{-1}\mathbf{C}_1)^{\top}.
\end{equation}
Given that $\mathbf{U}_t = \mathbf{R}_0\mathbf{Y}_t\mathbf{C}_0^\top - \mathbf{R}_1 \mathbf{Y}_{t-1} \mathbf{C}_1^\top$, the first-order
derivatives with respect to $\mathbf{R}_0, \mathbf{R}_1, \mathbf{C}_0$ and $\mathbf{C}_1$ are given by
\begin{equation}
\begin{split}
  \frac{\partial\ell(\cdot)}{\partial \mathbf{R}_0}& =n(T-1)\mathbf{R}_0^{-\top}-\sum_{t=2}^T \mathbf{U}_t \mathbf{C}_0 \mathbf{Y}_t^{\top}\\
  \frac{\partial\ell(\cdot)}{\partial \mathbf{R}_1}& =\sum_{t=2}^T \mathbf{U}_t \mathbf{C}_1 \mathbf{Y}_{t-1}^{\top}\\
  \frac{\partial\ell(\cdot)}{\partial \mathbf{C}_0}& =m(T-1)\mathbf{C}_0^{-\top}-\sum_{t=2}^T \mathbf{Y}_t^{\top} \mathbf{R}_0^\top \mathbf{U}_t\\
  \frac{\partial\ell(\cdot)}{\partial \mathbf{C}_1}&=\sum_{t=2}^T \mathbf{Y}_{t-1}^{\top} \mathbf{R}_1^\top \mathbf{U}_t.
\end{split}
\end{equation}
Since the first-order conditions for $\mathbf{C}_0$ and $\mathbf{R}_0$ do not have an analytic solution, the parameters can be estimated by
numerical optimization. However, this is hampered by the dimension of $\mathbf{C}_0$ and $\mathbf{R}_0$ which makes this approach unfeasible,
especially for high-dimensional problems. Alternatively, one can rely on the iterative estimation \citep{Chen2021} of
$\mathbf{A}_1, \mathbf{B}_1, \Sigmac, \Sigmar$, often referred to as the ``flip-flop'' algorithm \citep[see][]{Dutilleul1999}. In this case, only
one matrix is updated, while keeping fixed the others using the following iterations
\begin{equation}
\begin{split}
  \mathbf{A}_1 &\leftarrow \left(\sum_{t=2}^{T}\mathbf{Y}_t\Sigmac^{-1}\mathbf{B}_1\mathbf{Y}_{t-1}^\top\right)
  \left(\sum_{t=2}^{T}\mathbf{Y}_{t-1}\mathbf{B}_1^\top\Sigmac^{-1}\mathbf{B}_1\mathbf{Y}_{t-1}^\top\right)^{-1}\\
  \mathbf{B}_1 &\leftarrow \left(\sum_{t=2}^{T}\mathbf{Y}_t^\top \Sigmar^{-1}\mathbf{A}_1\mathbf{Y}_{t-1}\right)
  \left(\sum_{t=2}^{T}\mathbf{Y}_{t-1}^\top \mathbf{A}_1^\top \Sigmar^{-1}\mathbf{A}_1\mathbf{Y}_{t-1}\right)^{-1}\\
  \Sigmac &\leftarrow \frac{\sum_{t=2}^{T}\mathbf{E}_t^\top \Sigmar^{-1}\mathbf{E}_t}{m(T-1)}\\
  \Sigmar &\leftarrow \frac{\sum_{t=2}^{T}\mathbf{E}_t \Sigmac^{-1}\mathbf{E}_t^\top}{n(T-1)},
\end{split}
\end{equation}
where $\mathbf{A}_1$ and $\Sigmar$ have a unit Frobenius norm \citep[see][for details]{Chen2021,Bucci2026}. Since we impose
$\mathbf{R}_0 = \mathbf{I}_m$ in Section \ref{sec:identify}, here we focus only on $\mathbf{C}_0$. We can recover $\hat{\mathbf{C}}_0^\top$ as
the Cholesky factor of the estimated inverse of $\Sigmac$, \textit{i.e.},
\begin{equation}
  \label{eq:Cholesky}
  \hat{\bm{\Sigma}}_{\text{c}}^{-1} = \hat{\mathbf{C}}_0^\top\hat{\mathbf{C}}_0.
\end{equation}
For the purpose of
interpretability, we normalize the effect of the structural shocks on the errors of the reduced-form model as follows:
\begin{equation}\label{eq:Q}
    \hat{\mathbf{Q}} = \hat{\mathbf{C}}_0 \langle\hat{\mathbf{C}}_0\rangle^{-1}
\end{equation}
where $\langle\hat{\mathbf{C}}_0\rangle = \text{diag}(\hat{\mathbf{C}}_0)$.

Using the following total score for
$\bm{\beta} = \left[\text{vec}(\mathbf{A}_1)^\top, \text{vec}(\mathbf{B}_1)^\top, \text{vech}(\Sigmar)^\top, \text{vech}(\Sigmac)^\top\right]^\top$,
$s_T(\bm{\beta}) = \sum_{t=2}^{T}s_t(\bm{\beta})$ (where $\operatorname{vech}$ is the operator that stacks the lower triangular elements of a
matrix in a vector), derived in Appendix \ref{sec:appendixA}
\begin{equation}
  \label{eq:score}
  \mathbf{s}_T(\bm{\beta}) = \begin{bmatrix}
    \frac{\partial \ell(\cdot)}{\partial \text{vec}(\mathbf{A}_1)}\\
    \frac{\partial \ell(\cdot)}{\partial \text{vec}(\mathbf{B}_1)}\\
    \frac{\partial \ell(\cdot)}{\partial \text{vech}(\Sigmac)}\\
    \frac{\partial \ell(\cdot)}{\partial \text{vech}(\Sigmar)}\\
  \end{bmatrix} = \begin{bmatrix}
    \sum_{t=2}^{T}\left[ (\mathbf{Y}_{t-1}\mathbf{B}_1^\top \Sigmac^{-1}) \otimes \Sigmar^{-1} \right] \text{vec}(\mathbf{E}_t)\\
    \sum_{t=2}^{T}\left[(\mathbf{Y}_{t-1}^\top \mathbf{A}_1^\top \Sigmar^{-1}) \otimes \Sigmac^{-1} \right] \text{vec}(\mathbf{E}_t^\top)\\
    - \frac{m}{2}(T-1)\mathbf{D}_n^\top\text{vec}(\Sigmac^{-1})+ \frac{1}{2}\sum_{t=2}^{T}\mathbf{D}_n^\top(\Sigmac^{-1} \otimes \Sigmac^{-1})
    \text{vec}(\mathbf{E}_t^{\top} \Sigmar^{-1} \mathbf{E}_t) \\
    -\frac{n}{2}(T-1)\mathbf{D}_m^\top\text{vec}(\Sigmar^{-1}) + \frac{1}{2}\sum_{t=2}^{T}\mathbf{D}_m^\top
    \left(\Sigmar^{-1} \otimes \Sigmar^{-1}\right)\text{vec}\left(\mathbf{E}_t\Sigmac^{-1}\mathbf{E}_t^\top\right)
\end{bmatrix},
\end{equation} 
where $\mathbf{D}_n$ is the duplication matrix that projects from the $n$-dimensional space onto the $n(n+1)/2$-dimensional space (the same is
for $\mathbf{D}_m$), \textit{i.e.}, $\text{vec}(\mathbf{X})=\mathbf{D}_n \text{vech}(\mathbf{X})$. We can compute the asymptotic covariance
matrix for the reduced-form parameters with the Fisher's Information matrix
\[
\bm{\mathcal{I}}(\bm{\beta}) = \mathbb{E}\left[\mathbf{s}_t(\bm{\beta})\mathbf{s}_t(\bm{\beta})^\top\right]
\]
which, under regularity conditions, asymptotically converges to $-\mathbb{E}[\bm{\mathcal{H}}(\bm{\beta})]$, where $\bm{\mathcal{H}}(\bm{\beta})$ is the
Hessian matrix of the log-likelihood function. Therefore, evaluating the score at the optimum, we can estimate the asymptotic covariance matrix
using the outer-product-of-gradients (OPG) estimator
$\widehat{\text{Var}}(\hat{\bm{\beta}}) = \left[ \sum_{t=2}^T \mathbf{s}_t(\hat{\bm{\beta}})\mathbf{s}_t(\hat{\bm{\beta}})^\top \right]^{-1}$.

In our case, we need to find the standard errors of the structural impact parameters.  Since the structural matrix ${\mathbf{Q}}^{-\top}$ is a
non-linear transformation of the reduced-form parameters, the standard errors of $\hat{\mathbf{Q}}^{-\top}$ are computed via the delta
method\footnote{As a robust alternative to ensure validity in finite samples, standard errors can also be obtained through a residual-based
  bootstrap procedure, which accounts for the potential non-normality of the innovations $\mathbf{E}_t$.}. To this end, we define the
$n(n-1)/2$-dimensional vector containing the sub-diagonal elements of $\hat{\mathbf{Q}}^{-\top}$,
$\operatorname{vecl}(\hat{\mathbf{Q}}^{-\top})$. Proposition \ref{prop:jacobian} provides the Jacobian of $\text{vecl}(\hat{\mathbf{Q}}^{-\top})$
with respect to $\operatorname{vech}(\hat{\bm{\Sigma}}_c)$.
\begin{proposition}
  \label{prop:jacobian}
  Let $\hat{\bm{\Sigma}}_c$ be the ML estimate of the reduced-form column covariance matrix,
  $\hat{\mathbf{C}}_0 = \operatorname{chol}(\hat{\bm{\Sigma}}_c^{-1})$ its upper triangular Cholesky factor, and
  $\hat{\mathbf{Q}} = \hat{\mathbf{C}}_0 \langle\hat{\mathbf{C}}_0\rangle^{-1}$ the structural unit upper triangular matrix, where
  $\langle\hat{\mathbf{C}}_0\rangle = \operatorname{diag}(\hat{\mathbf{C}}_0)$. The Jacobian
  $\mathbf{J} = \frac{\partial\,\operatorname{vecl}(\hat{\mathbf{Q}}^{-\top})}{\partial\,\operatorname{vech}(\hat{\bm{\Sigma}}_c)^{\top}}$ is
\[
  \mathbf{J} = \mathbf{L}_n^s \left[ \bigl(\hat{\mathbf{Q}}^{-1} \otimes \hat{\mathbf{C}}_0^{-\top}\bigr) -\sum_{k=1}^n
    \bigl(\mathbf{H}_k \otimes \hat{\mathbf{C}}_0^{-\top}\mathbf{H}_k\bigr)\right] \mathbf{T}_n\,
  \hat{\mathbf{M}}^{-1}\mathbf{D}_n^+\,(\hat{\bm{\Sigma}}_{\text{c}}^{-1}\otimes\hat{\bm{\Sigma}}_{\text{c}}^{-1})\,\mathbf{D}_n,
\]
where
$\hat{\mathbf{M}} = \mathbf{D}_n^+[(\hat{\mathbf{C}}_0^\top\otimes\mathbf{I}_n)
+(\mathbf{I}_n\otimes\hat{\mathbf{C}}_0^\top)\mathbf{K}_{nn}]\mathbf{T}_n$, $\mathbf{H}_k=\bm{\eta}_k\bm{\eta}_k^\top$ (where $\bm{\eta}_k$ is
the $k$-th canonical basis vector of $\mathbb{R}^n$), $\mathbf{L}_n^s$ is the strict elimination matrix selecting only the below-diagonal
elements, $\mathbf{T}_n$ is the lower triangular selection matrix, and $\mathbf{D}_n$, $\mathbf{D}_n^+$, $\mathbf{K}_{nn}$ are the duplication
matrix, its Moore--Penrose inverse, and the commutation matrix, respectively.
\end{proposition}

\begin{proof}
  See Appendix \ref{sec:proof}.
\end{proof}

Applying the delta method, the asymptotic distribution of the structural parameters is given by
\begin{equation}
  \label{eq:C0dist}
  \sqrt{T-1}\left(\text{vecl}(\hat{\mathbf{Q}}^{-\top}) 
    - \text{vecl}(\mathbf{Q}^{-\top})\right) 
  \xrightarrow{d} \mathcal{N}\!\left(\mathbf{0},\; 
    \mathbf{J}\, \mathbf{V}_{\Sigmac}\, 
    \mathbf{J}^{\top}\right),
\end{equation}
where $\mathbf{V}_{\Sigmac}$ is the asymptotic covariance matrix of the MLE estimator of the $n(n+1)/2$ singular elements of $\bm{\Sigmac}$.

\subsection{Structural impulse response functions}
\label{sec:IRFs}

Our approach computes structural impulse response functions (IRFs) using an explicit structural decomposition, in contrast to \cite{Chen2021}.
In fact, in our SMAR model the reduced-form errors ($\mathbf{E}_t$) are generated by a set of underlying uncorrelated structural shocks
($\mathbf{U}_t$) and the link is provided by the structural error in Equation \eqref{eq:Et}. Once we identify $\mathbf{C}_0$, we obtain the
orthogonal shocks and thus can trace their causal impact through the IRFs. The orthogonalization is thus performed on the column-wise
relationships, as captured by the matrix $\mathbf{C}_0$.

To compute the IRFs, for the sake of simplicity, we consider the vectorized form of a stable MAR(1) model
\[
  \mathbf{y}_t = \bm\Phi_1 \mathbf{y}_{t-1} + \mathbf{e}_t,
\]
where $\mathbf{y}_t = \text{vec}(\mathbf{Y}_t)$, $\bm \Phi_1 = (\mathbf{B}_1 \otimes \mathbf{A}_1)$ are the $mn \times mn$ coefficient matrices,
$\mathbf{e}_t = \text{vec}(\mathbf{E}_t)$ is the $mn \times 1$ vector of reduced-form errors. The Wold representation of this model is given by
\begin{equation}
  \label{eq:MA}
    \mathbf{y}_t = \sum_{h=0}^{\infty} \bm\Psi_h \mathbf{e}_{t-h},
\end{equation}
where each $mn \times mn$ matrix $\bm\Psi_h$ represent the $h$-th impulse response of $\mathbf{y}_t$ to the reduced-form shocks
$\mathbf{e}_t$. These matrices can be computed recursively \citep{Lutkepohl1990}, with $\bm\Psi_0 = \mathbf{I}_{mn}$, and
\[
    \bm\Psi_h = \bm\Phi_1^h = (\mathbf{B}_1 \otimes \mathbf{A}_1)^h = (\mathbf{B}_1^h \otimes \mathbf{A}_1^h) \qquad \text{for } h \geq 1.
\]
In the SMAR model, we have that $\mathbf{E}_t =\mathbf{R}_0^{-1}\mathbf{U}_t \mathbf{C}_0^{-\top}$, which in vectorized form becomes
\[
  \mathbf{e}_t = \text{vec}(\mathbf{R}_0^{-1}\mathbf{U}_t \mathbf{C}_0^{-\top}) = \left[\mathbf{C}_0^{-1} \otimes \mathbf{R}_0^{-1}\right] \mathbf{u}_t,
\]
where $\mathbf{u}_t = \text{vec}(\mathbf{U}_t)$. Let $\mathbf{S} = \left[\mathbf{C}_0^{-1} \otimes \mathbf{R}_0^{-1}\right]$ be the matrix that
maps the structural shocks to the reduced-form shocks. Substituting it in Equation \eqref{eq:MA}, yields
\[
    \mathbf{y}_t = \sum_{h=0}^\infty \bm\Psi_h(\mathbf{S}\mathbf{u}_{t-h}) = \sum_{h=0}^\infty (\bm\Psi_h \mathbf{S})\mathbf{u}_{t-h}.
\]
The structural impulse response coefficient denotes the response of $y_{i,t+h}$ to a unit shock in $u_{j,t}$ which, from the MA version of the
vectorized MAR in Equation \eqref{eq:MA}, is equal to
\[
  \frac{\partial y_{i,t+h}}{\partial u_{j,t}} = \bm{\eta}_i^\top \left(\bm{\Psi}_h \mathbf{S}\right) \bm{\eta}_j,
\]
where $\bm{\eta}_i$ and $\bm{\eta}_j$ represent the $i$-th and $j$-th columns of the identity matrix $\mathbf{I}_{mn}$. Since we imposed
$\mathbf{R}_0 = \mathbf{I}_m$, in our model we have $\mathbf{S}=\mathbf{C}_0^{-1} \otimes \mathbf{I}_m.$

To compute the confidence intervals of the structural IRFs, we employ a residual bootstrap procedure \citep{Lutkepohl1990, Kilian1998}. Let
$\hat{\mathbf{E}}_t = \mathbf{Y}_t - \hat{\mathbf{A}}_1 \mathbf{Y}_{t-1} \hat{\mathbf{B}}_1^{\top}$ denote the matrix of reduced-form residuals
obtained from the estimated SMAR model, and let $\hat{\mathbf{C}}_0$ be the corresponding structural matrix, we obtain a bootstrap sequence of
innovations $\{\mathbf{E}_t^*\}_{t=1}^T$ by resampling from $\{\hat{\mathbf{E}}_t\}$ with replacement. Then, a bootstrap time series
$\{\mathbf{Y}_t^*\}_{t=1}^T$ is generated recursively using the estimated matrices of parameters as follows:
\begin{equation}
    \mathbf{Y}_t^* = \hat{\mathbf{A}}_1 \mathbf{Y}_{t-1}^* \hat{\mathbf{B}}_1^{\top} + \mathbf{E}_t^*,
\end{equation}
initializing at $\mathbf{Y}_0^* = \mathbf{Y}_0$. The SMAR model is then re-estimated on $\{\mathbf{Y}_t^*\}_{t=1}^{T}$ to obtain bootstrap
estimates $\hat{\mathbf{A}}_{1,i}^*$, $\hat{\mathbf{B}}_{1,i}^*$, and $\hat{\mathbf{C}}_0^*$, from which the corresponding structural IRFs,
$\bm{\Theta}_h^* = \bm{\Psi}_h^* \mathbf{S}^*$, are computed as in Equation~\eqref{eq:MA}. Repeating this procedure $V$ times yields the
bootstrap distribution $\{\bm{\Theta}_h^{*(\nu)}\}_{\nu=1}^V$, from which pointwise $(1-\alpha)$ confidence bands at horizon $h$ are constructed
as the empirical $\alpha/2$ and $1-\alpha/2$ quantiles \citep{Hall1992}.

Since each asset $i$ enters the Kronecker structure symmetrically and we do not impose heterogeneous restrictions across assets in
$\mathbf{R}_0 = \mathbf{I}_m$, we report the cross-sectional average of the structural IRF for each variable pair $(j \to k)$:
\begin{equation}
  \label{eq:avg_IRF}
    \bar{\theta}_{jk,h} = \frac{1}{m} \sum_{i=1}^{m}
    \bm{\eta}_{(k-1)m+i}^\top \left(\bm{\Psi}_h \mathbf{S}\right)
    \bm{\eta}_{(j-1)m+i},
\end{equation}
where the index $(k-1)m+i$ selects asset $i$ of variable $k$ in the $\operatorname{vec}(\mathbf{Y}_t)$ ordering. Averaging across assets exploits
the panel dimension of the data and yields a parsimonious summary of the common propagation mechanism shared by the $m$ assets.

%%%%%%%%%%%%%%%%%%%%%%%%%%%%%%%%%%%%%%%%%%%%%%%%%%%%%%%%%%%%%%%%%%%%%%%%%%%%%%%%%%%%%%%%%%%%%%%%%%%%%%%%%

\section{SMAR model for returns, volatilities and volumes}
\label{sec:SMARfinance}

In this section, we present the structural form, \textit{i.e.}, the constraints on the $\mathbf{C}_0$ matrix in our model. Although the
relationship between trading volumes, volatility, and returns is often complex and context-dependent, the structure of $\mathbf{C}_0$ is imposed
by combining financial theory and empirical evidence. First, we refer to the Efficient Market Hypothesis \citep{Fama1970}, which postulates that
asset prices rapidly incorporate all available information, thus making financial returns unpredictable. Second, we consider the MDH, which
establishes that both trading volume and volatility are influenced by a latent factor that typically corresponds to the flow of new information
arriving in the market. Therefore, when new information becomes available, market participants tend to increase their trading activity, which
often results in greater price fluctuations. However, return shocks have an immediate impact on both volatility and volume. Conversely,
simultaneous changes in volatility or trading volume are very unlikely to influence returns within the same period \citep{Whitelaw1994}.

Despite trading volume and volatility generally exhibit a positive correlation, the direction of causality is likely mixed. Some studies suggest
that increased trading activity increases volatility, demonstrating that volume acts as a proxy for information inflow \cite[see, among
others,][]{harris87,Lamoureux1990,lebaron92}. In reality, the relationship is likely bidirectional, with feedback effects reinforcing both
variables. The MDH \citep{Tauchen1983} explains that a new flow of information can simultaneously trigger a shock to price dispersion
and trading activity. Empirical evidence consistently supports this mechanism across different markets and sample periods
\citep{Karpoff1987,Gallant1992, Andersen1996}. Based on this hypothesis, we choose to treat volume as the most endogenous variable, since it
simultaneously responds to return shocks through leverage \citep{Black1976, Christie1982} and reacts to volatility shocks through the MDH
channel. This shock structure (from asset returns to volatility and then to trading volume) is consistent with the idea that information shocks
first influence prices, then generate a reassessment of risk reflected in volatility, and finally influence trading activity
\citep{Lamoureux1990, Bessembinder1993}.

To exemplify the structural representation, we discuss a two-asset case in which the matrix of time series is 
\begin{equation}
\mathbf{Y}_t =
\begin{bmatrix}
\log\tv{1,t} & \log\bp{1,t} & \ret{1,t}  \\
\log\tv{2,t} & \log\bp{2,t} & \ret{2,t}  \\
\end{bmatrix},
\end{equation}
where $\tv{i,t}$ is the trading volume of asset $i$ at time $t$. $\bp{i,t}$ is the bipower variation of
asset $i$ at time $t$:
\[
  \bp{i,t} = \frac{\pi}{2} \sum_{\tau= 1}^{N_t-1} |r_{i,\tau}|\cdot |r_{i,\tau+1}|
\]
where $r_{i,\tau}$ being the $\tau$-th intraday return of asset $i$, and $N_t$ is the number of intraday observations on the $t$-th day, that we
use as a proxy for volatility. In addition, daily returns are standardized through a GARCH(1,1) model as follows:
\[
  \ret{i,t} = \frac{\Delta \log p_{i,t}}{\hat{\sigma}_{i,t}}
\]
where $p_{i,t}$ is the $t$-th price of asset $i$ and $\hat{\sigma}_{i,t}$ is the conditional volatility estimated with the GARCH(1,1) model.

Given that the matrix $\mathbf{C}_0$ is an upper triangular \citep{kilian2017}, we normalize the effect of the structural shocks on the errors of
the reduced-form model by introducing the column-normalized lower triangular matrix
\[
  \mathbf{Q}^{-\top} = \mathbf{C}_0^{-\top}\langle\mathbf{C}_0\rangle =
\begin{bmatrix}
  1      & 0     & 0\\
  c_{21}  & 1     & 0\\
  c_{31}  & c_{32} & 1    
\end{bmatrix},
\]
and the row-normalized triangular matrix
\[
  \mathbf{W}^{-1} = \langle\mathbf{R}_0\rangle\mathbf{R}_0^{-1} =
  \begin{bmatrix}
    1 & 0\\
    \rho_{21} & 1\\
  \end{bmatrix},
\]
where $\langle\hat{\mathbf{R}}_0\rangle = \operatorname{diag}(\mathbf{R}_0)$. It follows that the error of the model in Equation
\eqref{eq:reduced} can be computed as
\[
  \mathbf{E}_t = \mathbf{R}_0^{-1}\mathbf{U}_t\mathbf{C}_0^{-\top} = \langle\mathbf{R}_0\rangle^{-1}\mathbf{W}^{-1} \mathbf{U}_t
  \mathbf{Q}^{-\top}\langle\mathbf{C}_0\rangle^{-1},
\]
that is
\[
  \begin{bmatrix}
    \e{e}{vol}{1,t} &  \e{e}{rbp}{1,t} & \e{e}{ret}{1,t}\\
    \e{e}{vol}{2,t} &  \e{e}{rbp}{2,t} & \e{e}{ret}{2,t}
  \end{bmatrix}=\langle\mathbf{R}_0\rangle^{-1}\begin{bmatrix}
    1 & 0\\
    \rho_{21} & 1\\
  \end{bmatrix}\begin{bmatrix}
    \e{u}{vol}{1,t} &  \e{u}{rbp}{1,t} & \e{u}{ret}{1,t}\\
    \e{u}{vol}{2,t} &  \e{u}{rbp}{2,t} & \e{u}{ret}{2,t}
  \end{bmatrix} \begin{bmatrix}
    1      & 0     & 0\\
    c_{21} & 1     & 0\\
    c_{31} & c_{32} & 1\\
  \end{bmatrix}\langle\mathbf{C}_0\rangle^{-1}.
\]
Defining $\tilde{\mathbf{E}}_t = \mathbf{W}^{-1} \mathbf{U}_t \mathbf{Q}^{-\top}$, which are the standardized errors of the reduced-form model,
we get
\[
  \begin{split}
    &\ee{vol}{1,t} = \e{u}{vol}{1,t} + \e{u}{rbp}{1,t} c_{21} + \e{u}{ret}{1,t}c_{31}\\
    &\ee{rbp}{1,t} = \e{u}{rbp}{1,t} + \e{u}{ret}{1,t} c_{32}\\
    &\ee{ret}{1,t} = \e{u}{ret}{1,t}
  \end{split}
\]
for the first row of $\tilde{\mathbf{E}}_t$, and
\begin{equation}\label{eq:system}
\begin{split}
&\ee{vol}{2,t} = (\rho_{21}\e{u}{vol}{1,t} +\e{u}{vol}{2,t}) + (\rho_{21}\e{u}{rbp}{1,t} + \e{u}{rbp}{2,t}) c_{21} + (\rho_{21}\e{u}{ret}{1,t} + \e{u}{ret}{2,t})c_{31}\\
&\ee{rbp}{2,t} = (\rho_{21}\e{u}{rbp}{1,t} +\e{u}{rbp}{2,t}) + (\rho_{21}\e{u}{ret}{1,t} + \e{u}{ret}{2,t})c_{32}\\
&\ee{ret}{2,t} = \rho_{21}\e{u}{ret}{1,t} +\e{u}{ret}{2,t}
\end{split}
\end{equation}
for the second row of $\tilde{\mathbf{E}}_t$.

This decomposition highlights the explanatory power of the model. Indeed, the error equations for the first asset (the $n$-dimensional vector
$\tilde{\mathbf{e}}_{1,t}$) show how structural shocks to a financial variable simultaneously impact other variables within the same asset. The
parameters $c_{ij}$ quantify the magnitude of these effects across variables within the same asset.

The error equations for the second asset ($\tilde{\mathbf{e}}_{2,t}$) concern contagion between assets. The parameter $\rho_{21}$ governs the
extent to which structural shocks originating in asset 1 are transmitted to asset 2. For example, the term $\rho_{21} \e{u}{ret}{1,t}+\e{u}{ret}{2,t}$
represents the total return shock affecting asset 2, which is a combination of its shock and the portion spilling over from asset 1. As already
stated at the end of Section \ref{sec:identify}, we set $\mathbf{R}_0 = \mathbf{I}_m$ to avoid any contagion effect between all
assets. Analytically, the system \eqref{eq:system} reduces to
\begin{equation}\label{eq:system2} 
    \begin{split}
& \ee{vol}{2,t} = \e{u}{vol}{2,t} + \e{u}{rbp}{2,t}c_{21} + \e{u}{ret}{2,t}c_{31}\\
& \ee{rbp}{2,t} = \e{u}{rbp}{2,t} + \e{u}{ret}{2,t}c_{32}\\
& \ee{ret}{2,t} = \e{u}{ret}{2,t}.
\end{split}
\end{equation}

%%%%%%%%%%%%%%%%%%%%%%%%%%%%%%%%%%%%%%%%%%%%%%%%%%%%%%%%%%%%%%%%%%%%%%%%%%%%%%%%%%%%%%%%%%%%%%%%%%%%%%%%%

\section{Empirical analysis}
\label{sec:Empirical_application}

Our sample consists of the daily prices and volumes of the assets comprising the Dow Jones Industrial Average (as of April 2025), collected
between January 2021 and February 2025 ($T=1059$ observations). The data are sourced from Bloomberg. The variables used in the analysis are
computed as shown in Section \ref{sec:SMARfinance}. Moreover, to avoid estimating the constant-value matrix\footnote{Through a preliminary
  analysis (available upon request from the authors) of the time series, we noted that the inclusion of a constant in the SMAR model can generate
  collinearity problems for a large number of assets, and can often lead to the estimation of singular matrices. This also affects the estimation
  time, which becomes increasingly longer as $m$ increases.}, all series are standardized. Based on the results of the Bayesian Information
Criterion (BIC), we specify a SMAR(1).

As further detailed in Section \ref{sec:Estimation}, the reduced-form MAR(1) defined as in Equation \eqref{eq:reduced} is estimated via maximum
likelihood. The model is stationary since the product of the spectral radii is
$\rho(\hat{\mathbf{A}}_1)\,\rho(\hat{\mathbf{B}}_1)\approx 0.989$ \citep[see][]{Chen2021}. Equation \eqref{eq:B1hat} reports the estimated
$\hat{\mathbf{B}}_1$, with corresponding standard errors provided in parentheses and bold values indicating a 5\% significance:
\begin{equation}
  \label{eq:B1hat}
  \hat{\mathbf{B}}_1=\begin{bmatrix*}[r]
    \underset{(0.0199)}{\bf 3.5765}  &  \underset{(0.0159)}{\bf -0.3975}  & \underset{(0.0114)}{\bf -0.0460}\\
    \underset{(0.0118)}{\bf 0.6307} & \underset{(0.0155)}{\bf 1.9885} & \underset{(0.0105)}{0.0072}\\
    \underset{(0.0146)}{0.0188} & \underset{(0.0180)}{-0.0122} & \underset{(0.0142)}{0.0041}
  \end{bmatrix*}.
\end{equation}
For instance, the first column shows the impact on financial variables from past volumes. Consistently with the literature on volatility
clustering and long-memory properties in market activity, both trading volume and volatility exhibit a high degree of persistence. We also
observe a significant cross-variable feedback between volume and volatility. Specifically, past trading volume has a positive and significant
impact on current volatility, supporting the idea that intense activity anticipates higher price dispersion. Conversely, past volatility has a
significant effect on current volume, suggesting a potential ``cooling off'' effect in trading activity following periods of high
uncertainty. Finally, the coefficients related to lagged returns in the equations for volume and volatility are relatively small, although the
negative impact on volume is statistically significant. This somehow provides evidence of asymmetric dynamics, where past price movements
slightly influence market participation. In line with the efficient market hypothesis discussed in Section \ref{sec:identify}, asset returns
appear to be largely independent of the lagged values of all variables.

Then, using the estimated covariance matrix $\hat{\bm{\Sigma}}_{\text{c}}$, we derive the inverse of the contemporaneous column-wise matrix with
normalized diagonal, $\hat{\mathbf{Q}}^{-\top}$, which takes the following values (with standard errors reported in parentheses and bold values
indicating a 5\% significance):
\begin{equation}
  \label{eq:C0hat}
  \hat{\mathbf{Q}}^{-\top}=
  \begin{bmatrix*}[r]
    1.0000 & 0.0000 & 0.0000\\
    \underset{(0.0051)}{\bf 0.5075} & 1.0000 & 0.0000\\
    \underset{(0.0024)}{\bf -0.0460} &\underset{(0.0020)}{\bf -0.0473} & 1.0000
  \end{bmatrix*}.
\end{equation}
The interpretation of the estimated inverse of $\hat{\mathbf{Q}}^\top$ in Equation \eqref{eq:C0hat} provides a clear and economically significant
causal ordering. The highest coefficient is 0.5075, which indicates that a one-standard-deviation structural shock to market volatility (RBP)
leads to a strong, significant, and immediate positive response in trading volumes. This finding suggests that unexpected changes in market
uncertainty are a primary driver of contemporaneous trading activity, as market participants rapidly adjust their positions in response.

Furthermore, we find that a structural shock to returns has a negative impact on both trading volume (-0.0460) and volatility (-0.0473). This is
consistent with the leverage effect \citep{Black1976}: negative return shocks are associated with an immediate increase in perceived risk. The
nearly equal magnitude of the two coefficients suggests that the information content of the return shocks is transmitted symmetrically to the
other variables, rather than being absorbed first by one and then by the other. Notably, both effects are small relative to the volatility-volume
coefficient ($0.5075$), reinforcing the view that uncertainty about future prices, rather than the return itself, is the dominant contemporaneous
driver of trading activity.

Then, we analyze the average inverse response functions in Figure \ref{fig:IRFs}. Looking at the diagonal (\textit{i.e.}, self-responses to a
shock), it can be noticed that the self-response to a shock in returns tends to zero immediately after one lag. This is in large part coherent
with the theory of market efficiency by \citet{Fama1970}. Conversely, both the self-responses of volume and RBP decay slowly, consistent with the
volatility clustering theory \citep{Engle1982, Andersen2001b} and long-memory of the trading volume \citep{Bessembinder1993, Lobato2000}.

\begin{figure}[htbp]
  \centering
  \caption{Average structural impulse response function}
  \label{fig:IRFs}
  \includegraphics[width=\textwidth]{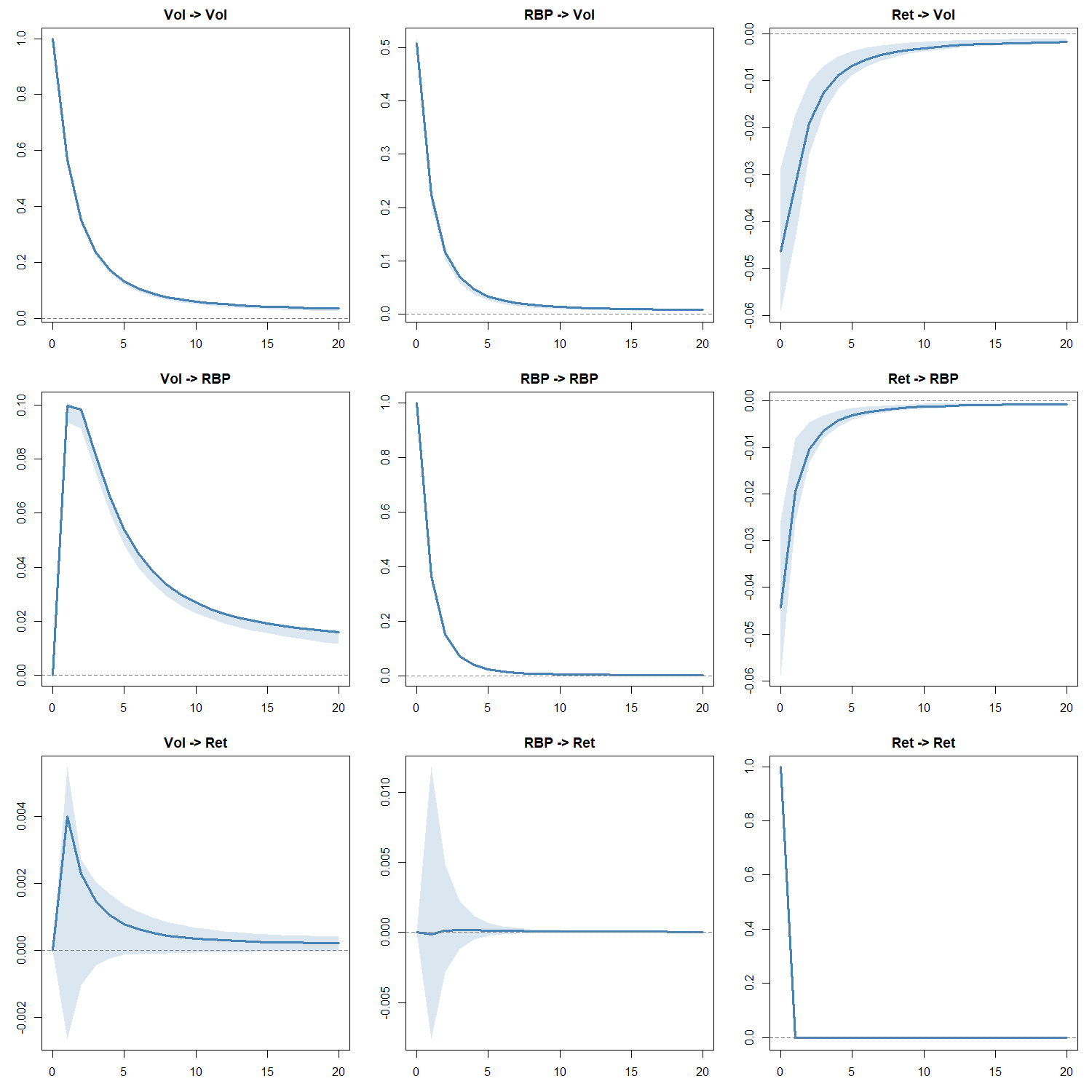} 
  \begin{flushleft}
    \note{The light blue area represents the 90\% bootstrapped confidence interval on 200 replicates. The first variable is the one shocked,
      while the second one is the response.}
  \end{flushleft}
\end{figure}

Interestingly, a shock in the realized bipower leads to an immediate and persistent increase in the volume. This finding confirms what is stated in
the MDH: news arrives and simultaneously moves prices and volumes. When we analyze the opposite direction (\textit{i.e.}, the effect of a shock
in the volumes on $\bp{}$), it emerges that the increase in the volatility is slightly lagged. Therefore, the intense trading activity precedes or
feeds the dispersion in prices in the following periods. It should be further noticed that the magnitude of the response of trading volume to a
shock in $\bp{}$ is markedly different from the opposite. This implies market participants massively
react increasing their trading activity when prices start fluctuating. In other words, volatility generates uncertainty, thus it forces buyers
and sellers to immediately adjust their investment strategies. In addition, we find that the effect of a shock in the volumes on volatility is
much more persistent than the opposite. This asymmetry in persistence constitutes, to our knowledge, a novel finding in the structural analysis
of the volume--volatility nexus.

The existing literature has primarily documented the contemporaneous and short-run relationship between volumes and
volatility. \citet{Gallant1992}, using a semi-nonparametric bivariate framework, provide impulse response profiles for stock prices and volume
and document that volume reacts strongly to price innovations, but do not examine the dynamic propagation between trading volume and volatility,
nor distinguish the persistence of shocks across these two variables. \citet{Tauchen1983} and \citet{Andersen1996} focus on the contemporaneous
co-movement implied by the latent information flow under the MDH, without characterizing the differential decay of cross-variable impulse
responses at medium horizons. \citet{Bollerslev1999}, examining the common fractional integration and long-run dependence between volume and
volatility, find that both variables share a common fractionally integrated component driven by information arrival, but their framework
abstracts from the short-to-medium-run asymmetry in impulse propagation that we document here.

The key contribution of our structural identification is that, by orthogonalizing shocks via $\hat{\mathbf C}_0$, we can isolate the dynamic
effect of structural volume innovations on volatility. The resulting persistence is consistent with market microstructure models in which
information is incorporated gradually through order flow \citep{Kyle1985}, implying that the impact of trading activity on price variability
unfolds over several trading days rather than instantaneously.

\subsection{Decomposition of Trading Volume Variance}

One of the aims of this paper is to disentangle the component of trading volume driven by informational shocks from that attributable to
liquidity dynamics. The structural framework introduced in Section \ref{sec:Model} makes this decomposition possible through the forecast error
variance decomposition (FEVD) of the model.

Since the SMAR operates on an $m \times n$ data matrix time series, the forecast error variance of trading volume for any given asset $i$
receives contributions from two distinct sources: shocks originating within the same asset (column-wise effects, governed by
$\mathbf{C}_0$) and shocks originating in other assets (row-wise effects, transmitted through the lagged dynamics of
$\hat{\mathbf{A}}_1$). We refer to these as the within-asset and cross-asset components, respectively.

Formally, we derive the FEVD from the structural moving average representation established in Section~\ref{sec:IRFs}. Recall that the vectorized
SMAR admits the representation
\begin{equation}
  \label{eq:SMA}
  \mathbf{y}_t = \sum_{l=0}^{\infty} \boldsymbol{\Theta}_l\, \mathbf{u}_{t-l},
  \qquad \boldsymbol{\Theta}_l \equiv \boldsymbol{\Psi}_l \mathbf{S},
\end{equation}
where $\mathbf{S} = \left[\mathbf{C}_0^{-1} \otimes \mathbf{I}_m\right]$ (using $\mathbf{R}_0=\mathbf{I}_m$),
$\mathbf{u}_t=\mathrm{vec}(\mathbf{U}_t)$ is the $mn\times 1$ vector of structural shocks with
$\bm{\Omega} = \mathrm{Var}(\mathbf{u}_t)=\mathbf{I}_{mn}$ (see Equation~\eqref{eq:SigmaD}), and the $mn\times mn$ matrices $\boldsymbol{\Psi}_l$
are the reduced-form impulse responses computed recursively as in Section~\ref{sec:IRFs}.

The $h$-step-ahead forecast error \citep{lutkepohl2005}, conditional on information at time $t$, is
\begin{equation}
  \label{eq:fe}
    \boldsymbol{\xi}_{t+h|t} \equiv \mathbf{y}_{t+h} - \mathbb{E}_t\!\left[\mathbf{y}_{t+h}\right]
    = \sum_{l=0}^{h-1} \boldsymbol{\Theta}_l\, \mathbf{u}_{t+h-l}.
\end{equation}
Let $\iota_{ik}\equiv(k-1)m+i$ denote the index of variable $k$ of asset $i$ in $\mathbf{y}_t$, and let $\bm{\eta}_s$ denote the $s$-th column of
$\mathbf{I}_{mn}$, so that $\bm{\eta}_{\iota_{ik}}^\top\mathbf{y}_t$ selects the element corresponding to variable $k$ of asset $i$. Since the
structural shocks are mutually uncorrelated with unit variance, the $h$-step forecast error variance of variable $k$ for asset $i$ is
\begin{align}
  \mathcal{V}_{ik}(h)
  &\equiv \mathrm{Var}\!\left(\bm{\eta}_{\iota_{ik}}^\top\,\boldsymbol{\xi}_{t+h|t}\right)
    = \bm{\eta}_{\iota_{ik}}^\top
    \!\left(\sum_{l=0}^{h-1}\boldsymbol{\Theta}_l\,\mathbf{I}_{mn}\,\boldsymbol{\Theta}_l^\top\right)
    \bm{\eta}_{\iota_{ik}} \notag\\
  &= \sum_{l=0}^{h-1}\bm{\eta}_{\iota_{ik}}^\top\boldsymbol{\Theta}_l\boldsymbol{\Theta}_l^\top\bm{\eta}_{\iota_{ik}}. \label{eq:MSE}
\end{align}
Since
$\bm{\eta}_{\iota_{ik}}^\top\boldsymbol{\Theta}_l\boldsymbol{\Theta}_l^\top\bm{\eta}_{\iota_{ik}} =
\left\|\bm{\eta}_{\iota_{ik}}^\top\boldsymbol{\Theta}_l\right\|^2 =
\sum_{s=1}^{mn}\!\left[\bm{\eta}_{\iota_{ik}}^\top\boldsymbol{\Theta}_l\,\bm{\eta}_s\right]^2$, Equation~\eqref{eq:MSE} can be written as
\begin{equation}
  \label{eq:MSEexpanded}
  \mathcal{V}_{ik}(h) = \sum_{l=0}^{h-1}\sum_{s=1}^{mn}
  \left[\bm{\eta}_{\iota_{ik}}^\top\boldsymbol{\Theta}_l\,\bm{\eta}_s\right]^2,
\end{equation}
where each term $\left[\bm{\eta}_{\iota_{ik}}^\top\boldsymbol{\Theta}_l\,\bm{\eta}_s\right]^2$ represents the squared structural impulse response
at horizon $l$. Because the $mn$ structural shocks are orthogonal (\textit{i.e.}, $\bm{\Omega}=\mathbf{I}_{mn}$), the total variance decomposes
additively \citep{lutkepohl2005, kilian2017}. Denoting the contribution of the $s$-th structural shock as
\begin{equation}
  \label{eq:Vsingle}
    \mathcal{V}_{ik}^{(s)}(h) \equiv \sum_{l=0}^{h-1}\left[\bm{\eta}_{\iota_{ik}}^\top\boldsymbol{\Theta}_l\,\bm{\eta}_s\right]^2,
\end{equation}
it follows immediately from~\eqref{eq:MSEexpanded} that
\[
    \sum_{s=1}^{mn}\mathcal{V}_{ik}^{(s)}(h) = \mathcal{V}_{ik}(h).
\]
In $\mathbf{u}_t$, the structural shocks originating from asset $i$ occupy the $n$ positions $\left\{(j-1)m+i\colon j=1,\ldots,n\right\}$,
corresponding to the three financial variables ($\tv{},\bp{},\ret{}$) of that asset. Summing over these positions gives the within-asset
component
\begin{align}
    \mathcal{V}_{ik}^{\mathrm{within}}(h)
    &= \sum_{j=1}^{n}\mathcal{V}_{ik}^{((j-1)m+i)}(h)
     = \sum_{j=1}^{n}\sum_{l=0}^{h-1}
       \left[\bm{\eta}_{\iota_{ik}}^\top\boldsymbol{\Theta}_l\,\bm{\eta}_{(j-1)m+i}\right]^2,
    \label{eq:within}
\end{align}
which we denote $\mathcal{V}_{ik}^{j}(h)\equiv\mathcal{V}_{ik}^{((j-1)m+i)}(h)$ for $j\in\{1,\ldots,n\}$. The cross-asset component
collects all remaining shock contributions and is obtained by difference
\begin{equation}
    \label{eq:cross}
    \mathcal{V}_{ik}^{\mathrm{cross}}(h) = \mathcal{V}_{ik}(h) - \mathcal{V}_{ik}^{\mathrm{within}}(h)
    = \sum_{\substack{s=1\\s\,\notin\,\{(j-1)m+i\}_{j=1}^{n}}}^{mn}
      \mathcal{V}_{ik}^{(s)}(h).
\end{equation}
Normalising by the total variance and averaging across the $m$ assets yields
\begin{equation}\label{eq:FEVD}
    \overline{\mathrm{FEVD}}_{jk}(h) =
    \frac{\displaystyle\frac{1}{m}\sum_{i=1}^{m}\mathcal{V}_{ik}^{j}(h)}
         {\displaystyle\frac{1}{m}\sum_{i=1}^{m}\mathcal{V}_{ik}(h)},
    \qquad j \in \{\tv{},\,\bp{},\,\ret{}\},
\end{equation}
and analogously for the cross-asset share
\begin{equation}
  \label{eq:FEVDcross}
  \overline{\mathrm{FEVD}}_{\mathrm{cross},k}(h) =
  \frac{\displaystyle\frac{1}{m}\sum_{i=1}^{m}\mathcal{V}_{ik}^{\mathrm{cross}}(h)}
  {\displaystyle\frac{1}{m}\sum_{i=1}^{m}\mathcal{V}_{ik}(h)}.
\end{equation}
By
construction,$\overline{\mathrm{FEVD}}_{\tv{\,},k}(h) + \overline{\mathrm{FEVD}}_{\bp{\,},k}(h) + \overline{\mathrm{FEVD}}_{\ret{\,},k}(h) +
\overline{\mathrm{FEVD}}_{\mathrm{cross},k}(h) = 1$ at every horizon $h$. At $h=1$, since $\mathbf{R}_0 = \mathbf{I}_m$ rules out any
contemporaneous cross-asset transmission, $\overline{\mathrm{FEVD}}_{\mathrm{cross},k}(1) = 0$ exactly and the three within-asset components
account for the full variance. As the horizon grows, cross-asset contributions accumulate through the lagged dynamics encoded in
$\hat{\mathbf{A}}_1$.

We label $\overline{\mathrm{FEVD}}_{\bp{},k}$ as the \textit{informative component} of trading volume, since it captures the fraction of
trading activity driven by the latent information flow that simultaneously moves prices and volumes under the MDH
\citep{Tauchen1983}. Conversely, $\overline{\mathrm{FEVD}}_{\tv{},k}$ constitutes the \textit{liquidity component}, reflecting trading
activity unrelated to contemporaneous changes in price dispersion.

Table~\ref{tab:FEVD} reports the four components for trading volume ($k = \tv{}$) at horizons $h \in \{1, 5, 10, 20\}$. At the impact
horizon ($h=1$), approximately 20.4\% of volume variance is attributable to the RBP shock and 79.4\% to the volume's own shock, confirming that
same-day trading activity is dominated by asset-specific liquidity dynamics with a non-trivial informational contribution. The return shock
accounts for a negligible 0.2\%, reflecting the near-zero contemporaneous effect of returns on volume visible in Figure~\ref{fig:IRFs}.

\begin{table}[ht]
\centering
\caption{Average FEVD of trading volume (\%)}
\label{tab:FEVD}
\begin{tabular}{lrrrr}
\toprule
Shock & $h=1$ & $h=5$ & $h=10$ & $h=20$ \\
\midrule
$\bp{t}$ (informative)   & 20.4 & 13.1 &  9.9 &  7.2 \\
$\tv{t}$ (liquidity)     & 79.4 & 61.8 & 47.9 & 35.8 \\
$\ret{t}$                &  0.2 &  0.1 &  0.1 &  0.1 \\
Cross-Asset Spillovers   &  0.0 & 25.0 & 42.1 & 56.9 \\
\bottomrule
\end{tabular}
\end{table}

As the horizon expands, both within-asset components decline in favour of cross-asset spillovers. At $h=20$, cross-asset dynamics explain 56.9\%
of total volume variance, while the informative and liquidity components decline to 7.2\% and 35.8\%, respectively. This decrease reflects the
high integration of the Dow Jones constituents: after the initial within-asset response, volatility and liquidity shocks in one stock propagate
to others through the estimated $\hat{\mathbf{A}}_1$ matrix, generating system-wide co-movement in trading activity that dwarfs the asset-level
components at medium and long horizons. The stacked area chart in Figure~\ref{fig:FEVD_plot} visualizes this transition continuously over the
20-day horizon.

\begin{figure}[h!]
  \centering
  \caption{Stacked FEVD of Trading Volume}
  \label{fig:FEVD_plot}
  \par\smallskip
  \includegraphics[scale=0.7]{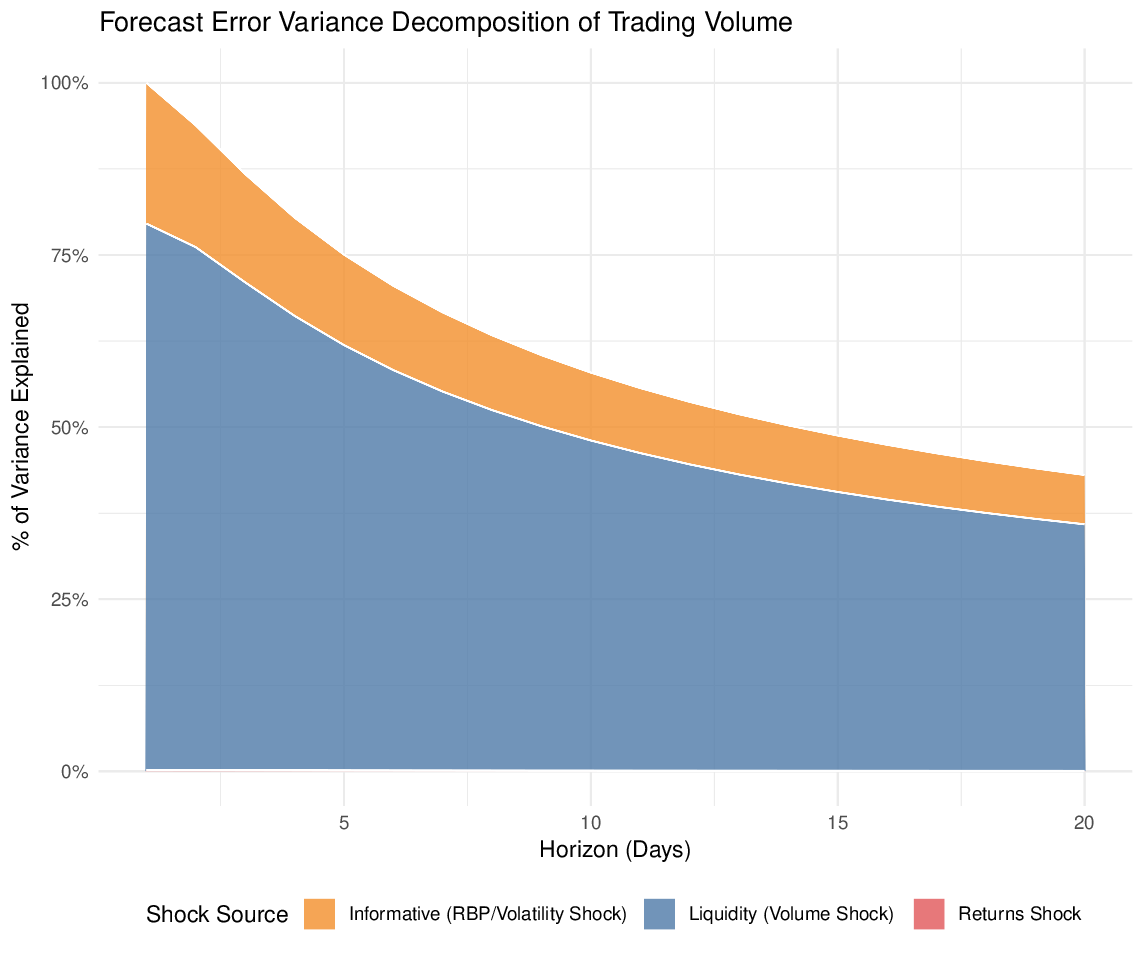}
    \note{The blue area represents the liquidity component (own-volume shock), the orange area represents the informative
    component (RBP shock), while the red area identifies the fraction of variance explained by returns. The white space above the colored areas
    indicates the variance explained by cross-asset spillovers.}
\end{figure}

\subsection{Informative component around FOMC}

To further validate our structural identification, we analyze the time series dynamics of the informative component around exogenous information
shocks. We define the daily informative share at time $t$ as
\begin{equation}\label{eq:daily_share}
  \widehat{s}_t = \frac{\displaystyle\sum_{i=1}^{m}
    \left(\hat{c}_{21}\, \hat{u}_{i,t}^{\textsc{rbp}}\right)^2}
  {\displaystyle\sum_{i=1}^{m}\left(\tilde{\hat{e}}_{i,t}^{\textsc{vol}}\right)^2},
\end{equation}
where $\tilde{\hat{e}}_{i,t}^{\textsc{vol}}$ is the standardized reduced-form residual for the volume of asset $i$ at day $t$,
$\hat{u}_{i,t}^{\textsc{rbp}} = \sum_{j} [\hat{\mathbf{Q}}^\top]_{j,\textsc{rbp}}\, \tilde{\hat{e}}_{i,j,t}$ is the recovered structural RBP
shock obtained by inverting $\tilde{\mathbf{E}}_t = \mathbf{U}_t\mathbf{Q}^{-\top}$, and $\hat{c}_{21}$ is the estimated element of
$\mathbf{Q}^{-\top}$ governing the contemporaneous impact of volatility on volume. This sample analog of the FEVD informative component varies
daily with the realized residuals and provides a time-varying measure of the information content of trading activity.

We then conduct an event study centered on Federal Open Market Committee (FOMC) announcements from 2021 to 2025 \citep{Lucca2015}. As shown in
Figure \ref{fig:EventStudy}, the informative share of trading volume remains stable and close to the unconditional mean (dotted line) during the
pre-announcement window, followed by a sharp and statistically significant spike exactly on the day of the event ($t=0$). This immediate reaction
confirms that our SMAR-based decomposition effectively isolates volume surges triggered by systematic information arrival. Furthermore, the rapid
reversion toward the mean in the subsequent days ($t+1, t+2$) suggests that the Dow Jones constituents incorporate macroeconomic news
efficiently, validating the RBP-volume channel as a reliable proxy for the latent information flow.

\begin{figure}[htbp]
  \centering
  \caption{Informative share of volume around FOMC announcements (Validated)}
  \label{fig:EventStudy}
  \includegraphics[scale=0.45]{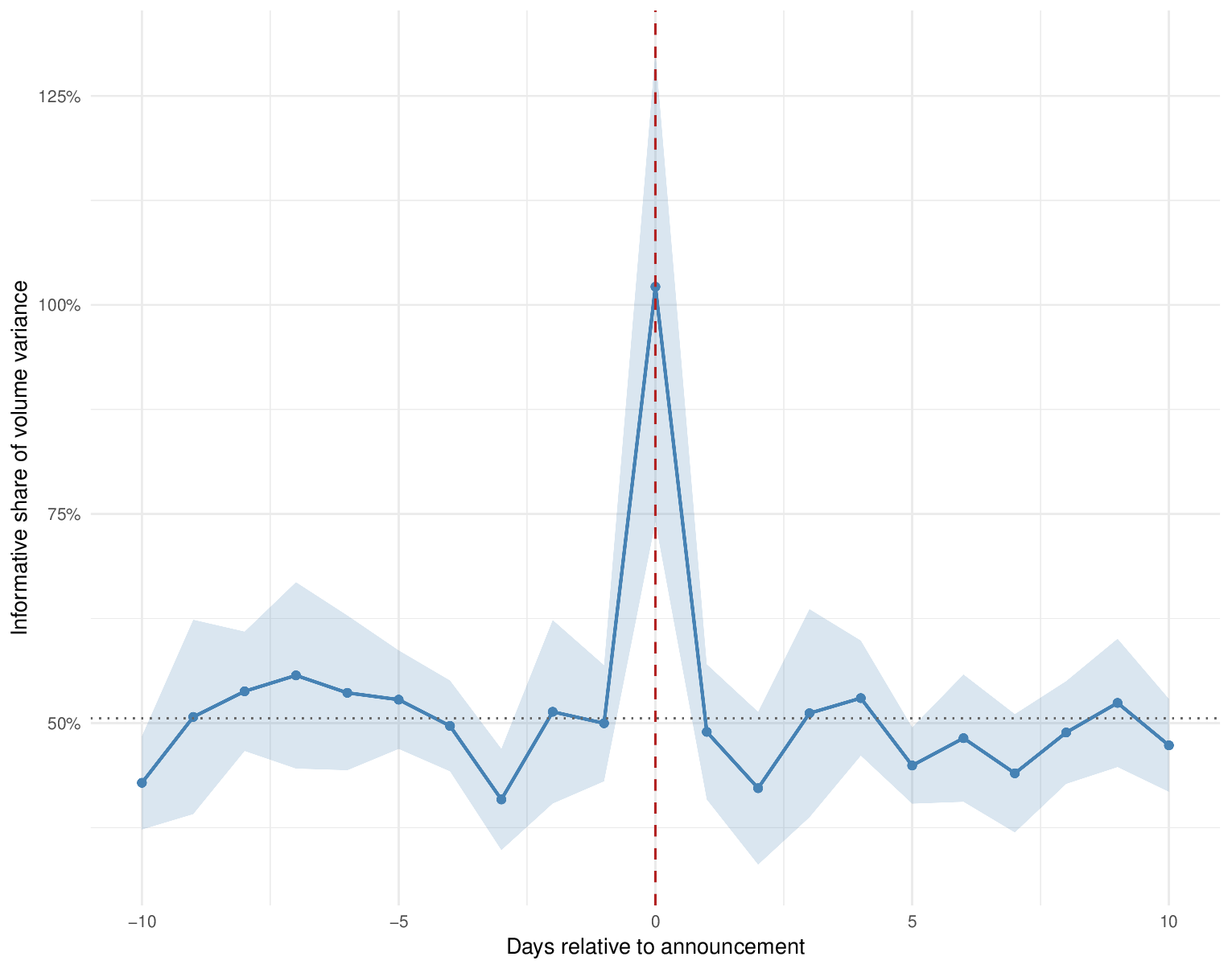} 
  \note{The plot shows the average daily informative share of volume variance across 33 FOMC events. The spike at $t=0$ denotes the
    contemporaneous impact of the macro announcement. The dotted line represents the sample mean, and the shaded area indicates the 90\%
    confidence interval.}
\end{figure}

\section{Conclusions}\label{sec:Conclusions}

This paper introduces the Structural Matrix Autoregressive (SMAR) model, a flexible framework that jointly accounts for contemporaneous and
dynamic interdependencies in high-dimensional financial time series. By extending the MAR model of \cite{Chen2021}, the SMAR explicitly
incorporates structural relationships between assets and financial variables via two contemporaneous interaction matrices, $\mathbf{C}_0$ and
$\mathbf{R}_0$, thereby providing a tractable identification strategy for orthogonal structural shocks and a natural decomposition of forecast
error variance even in large cross-sections.

The identification of the structural matrix $\mathbf{C}_0$ is based on a theory-driven recursive ordering, \textit{i.e.}, returns, realized
bipower variance, and trading volume, consistent with the efficient market hypothesis and the mixture-of-distributions hypothesis. The assumption
$\mathbf{R}_0 = \mathbf{I}_m$, while tractable, abstracts from contemporaneous cross-asset spillovers and represents a natural departure point
for future extensions.

Applied to daily data for the constituents of the Dow Jones Industrial Average from January 2021 to February 2025, the SMAR delivers four main
findings. First, the estimated contemporaneous structure confirms that unexpected changes in realized bipower variance are the primary driver of
contemporaneous trading activity: a one-standard-deviation volatility shock results in an immediate and economically large increase in volumes,
consistent with the MDH. Second, the structural impulse responses reveal a pronounced asymmetry. In fact, while volatility shocks generate a
large and rapid volume response, the effect of volume shocks on volatility is comparatively modest at short horizons but markedly more persistent
at longer horizons, suggesting that liquidity-driven trading gradually feeds back into price dispersion rather than producing an immediate
reassessment of risk. Third, the forecast error variance decomposition shows that at the one-day horizon trading volume is mainly driven by
own-liquidity dynamics (79.4\%) with a non-trivial informational component (20.4\%); as the horizon extends to twenty days, cross-asset
spillovers account for over half of total volume variance (56.9\%), highlighting the strong integration of Dow Jones constituents. Fourth, an
event study around Federal Open Market Committee announcements confirms that the $\bp{}$-based informative component spikes sharply and significantly
on announcement days and reverts rapidly thereafter.

Despite the promising results of the SMAR, several challenges remain. For instance, the assumption of no contemporaneous cross-asset effects is
an oversimplification that could be relaxed. Future work could explore more complex identification strategies for $\mathbf{R}_0$, \textit{e.g.},
using network theory or factor models. Additionally, the framework could be extended to incorporate time-varying parameters, allowing the model
to capture evolving market structures and changing responses to different economic regimes.

%%%%%%%%%%%%%%%%%%%%%%%%%%%%%%%%%%%%%%%%%%%%%%%%%%%%%%%%%%%%%%%%%%%%%%%%%%%%%%%%%%%%%%%%%%%%%%%%%%%%%%%%% 

\bibliography{references}
\clearpage
\appendix

\section{Derivation of the score function}
\label{sec:appendixA}

We derive the per-observation score $\mathbf{s}_t(\bm{\beta})$ from the log-likelihood contribution
\[
  \ell_t = -\frac{m}{2}\log|\Sigmac|-\frac{n}{2}\log|\Sigmar|-\frac{1}{2}\operatorname{tr}\!\left(\Sigmar^{-1}\mathbf{E}_t\Sigmac^{-1}\mathbf{E}_t^\top\right),
\]
where $\mathbf{E}_t = \mathbf{Y}_t - \mathbf{A}_1\mathbf{Y}_{t-1}\mathbf{B}_1^\top$ is $m\times n$, $\Sigmar$ is $m\times m$, and $\Sigmac$ is
$n\times n$.  We rely on two standard tools throughout. First, the trace inner product: $d f = \operatorname{tr}(\mathbf{Z}\,d\mathbf{X})$
identifies $\partial f/\partial \mathbf{X} = \mathbf{Z}^\top$, and equivalently
$df = [\operatorname{vec}(\mathbf{Z}^\top)]^\top \operatorname{vec}(d\mathbf{X})$. Second, the duplication matrix $\mathbf{D}_k$ of order $k$,
defined by $\operatorname{vec}(\bm{\Upsilon}) = \mathbf{D}_k\,\operatorname{vech}(\bm{\Upsilon})$ for any symmetric $k\times k$ matrix
$\bm{\Upsilon}$, so that restricting the differential to the symmetric subspace gives
$\operatorname{vec}(d\bm{\Upsilon}) = \mathbf{D}_k\,d\operatorname{vech}(\bm{\Upsilon})$, and therefore
$\partial f/\partial\operatorname{vech}(\bm{\Upsilon}) = \mathbf{D}_k^\top\operatorname{vec}(\partial f/\partial \bm{\Upsilon})$.

\subsection{Score with respect to $\text{vec}(\mathbf{A}_1$)}

Taking the differential of the quadratic term with $d\mathbf{E}_t = -(d\mathbf{A}_1)\mathbf{Y}_{t-1}\mathbf{B}_1^\top$ gives the following
\[
    d\ell_t
    = -\frac{1}{2}\operatorname{tr}\!\left(\Sigmar^{-1}d\mathbf{E}_t\,\Sigmac^{-1}\mathbf{E}_t^\top\right)
      -\frac{1}{2}\operatorname{tr}\!\left(\Sigmar^{-1}\mathbf{E}_t\,\Sigmac^{-1}d\mathbf{E}_t^\top\right).
\]
Because the two terms are transposes of each other inside the trace, they are equal, so
\[
    d\ell_t
    = \operatorname{tr}\!\left(\Sigmar^{-1}(d\mathbf{A}_1)\mathbf{Y}_{t-1}\mathbf{B}_1^\top\Sigmac^{-1}\mathbf{E}_t^\top\right)
    = \operatorname{tr}\!\left(\mathbf{Y}_{t-1}\mathbf{B}_1^\top\Sigmac^{-1}\mathbf{E}_t^\top\Sigmar^{-1}\,d\mathbf{A}_1\right),
\]
where the second step uses the cyclic property of the trace. Identifying
$\bm{\Xi}_{\mathbf{A}_1}= \mathbf{Y}_{t-1}\mathbf{B}_1^\top\Sigmac^{-1}\mathbf{E}_t^\top\Sigmar^{-1}$ yields
\[
    \frac{\partial\ell_t}{\partial\mathbf{A}_1} = \bm{\Xi}_{\mathbf{A}_1}^\top = \Sigmar^{-1}\mathbf{E}_t\Sigmac^{-1}\mathbf{B}_1\mathbf{Y}_{t-1}^\top,
\]
and therefore
\[
    \frac{\partial\ell_t}{\partial\operatorname{vec}(\mathbf{A}_1)} = \operatorname{vec}\!\left(\Sigmar^{-1}\mathbf{E}_t\Sigmac^{-1}\mathbf{B}_1\mathbf{Y}_{t-1}^\top\right)
    = \left[(\mathbf{Y}_{t-1}\mathbf{B}_1^\top\Sigmac^{-1})\otimes\Sigmar^{-1}\right]\operatorname{vec}(\mathbf{E}_t),
\]
where the last equality uses $\operatorname{vec}(\mathbf{PQR}) = (\mathbf{R}^\top\otimes \mathbf{P})\operatorname{vec}(\mathbf{Q})$.

\subsection{Score with respect to $\text{vec}(\mathbf{B}_1$)}
Now $d\mathbf{E}_t = -\mathbf{A}_1\mathbf{Y}_{t-1}(d\mathbf{B}_1)^\top$, hence $d\mathbf{E}_t^\top = -(d\mathbf{B}_1)\mathbf{Y}_{t-1}^\top\mathbf{A}_1^\top$.  Arguing as above:
\[
    d\ell_t
    = \operatorname{tr}\!\left(\Sigmar^{-1}\mathbf{E}_t\Sigmac^{-1}(d\mathbf{B}_1)\mathbf{Y}_{t-1}^\top\mathbf{A}_1^\top\right)
    = \operatorname{tr}\!\left(\mathbf{Y}_{t-1}^\top\mathbf{A}_1^\top\Sigmar^{-1}\mathbf{E}_t\Sigmac^{-1}\,d\mathbf{B}_1\right).
\]
Identifying $\bm{\Xi}_{\mathbf{B}_1} = \mathbf{Y}_{t-1}^\top\mathbf{A}_1^\top\Sigmar^{-1}\mathbf{E}_t\Sigmac^{-1}$ yields
\[
    \frac{\partial\ell_t}{\partial\mathbf{B}_1} = \bm{\Xi}_{\mathbf{B}_1}^\top = \Sigmac^{-1}\mathbf{E}_t^\top\Sigmar^{-1}\mathbf{A}_1\mathbf{Y}_{t-1},
\]
and therefore
\[
  \frac{\partial\ell_t}{\partial\operatorname{vec}(\mathbf{B}_1)} =
  \operatorname{vec}\!\left(\Sigmac^{-1}\mathbf{E}_t^\top\Sigmar^{-1}\mathbf{A}_1\mathbf{Y}_{t-1}\right) =
  \left[(\mathbf{Y}_{t-1}^\top\mathbf{A}_1^\top\Sigmar^{-1})\otimes\Sigmac^{-1}\right]\operatorname{vec}(\mathbf{E}_t^\top).
\]

\subsection{Score with respect to $\text{vech}(\mathbf{\Sigma}_c$)}

The matrices $\mathbf{A}_1$, $\mathbf{B}_1$, and $\Sigmar$ are held fixed.  Using $d(\Sigmac^{-1}) = -\Sigmac^{-1}(d\Sigmac)\Sigmac^{-1}$:
\begin{align*}
  d\ell_t
  &= -\frac{m}{2}\operatorname{tr}\!\left(\Sigmac^{-1}\,d\Sigmac\right)
    +\frac{1}{2}\operatorname{tr}\!\left(\Sigmar^{-1}\mathbf{E}_t\Sigmac^{-1}(d\Sigmac)\Sigmac^{-1}\mathbf{E}_t^\top\right)\\
  &= \operatorname{tr}\!\left(\underbrace{\left[-\frac{m}{2}\Sigmac^{-1}+\frac{1}{2}\Sigmac^{-1}\mathbf{E}_t^\top\Sigmar^{-1}\mathbf{E}_t\Sigmac^{-1}\right]}_{\displaystyle\mathbf{G}_c}\,d\Sigmac\right),
\end{align*}
where we used the cyclic property to bring $\Sigmar^{-1}\mathbf{E}_t$ to the right.  The matrix $\mathbf{G}_c$ is symmetric, so restricting to
the symmetric subspace via the duplication matrix gives
\[
  \frac{\partial\ell_t}{\partial\operatorname{vech}(\Sigmac)} = \mathbf{D}_n^\top\operatorname{vec}(\mathbf{G}_c) =
  -\frac{m}{2}\mathbf{D}_n^\top\operatorname{vec}(\Sigmac^{-1})
  +\frac{1}{2}\mathbf{D}_n^\top(\Sigmac^{-1}\otimes\Sigmac^{-1})\operatorname{vec}(\mathbf{E}_t^\top\Sigmar^{-1}\mathbf{E}_t),
\]
where the last equality uses $\operatorname{vec}(\Sigmac^{-1} \mathbf{Q} \Sigmac^{-1}) = (\Sigmac^{-1}\otimes\Sigmac^{-1})\operatorname{vec}(\mathbf{X})$ with $\mathbf{X} = \mathbf{E}_t^\top\Sigmar^{-1}\mathbf{E}_t$.

\subsection{Score with respect to $\text{vech}(\mathbf{\Sigma}_r$)}

An entirely symmetric argument, with the roles of $\Sigmar$ and $\Sigmac$ exchanged and $n$ replaced by $m$, gives
\[
    \mathbf{G}_r = -\frac{n}{2}\Sigmar^{-1}+\frac{1}{2}\Sigmar^{-1}\mathbf{E}_t\Sigmac^{-1}\mathbf{E}_t^\top\Sigmar^{-1},
\]
and therefore
\[
    \frac{\partial\ell_t}{\partial\operatorname{vech}(\Sigmar)}
    = \mathbf{D}_m^\top\operatorname{vec}(\mathbf{G}_r)
    = -\frac{n}{2}\mathbf{D}_m^\top\operatorname{vec}(\Sigmar^{-1})
      +\frac{1}{2}\mathbf{D}_m^\top(\Sigmar^{-1}\otimes\Sigmar^{-1})\operatorname{vec}(\mathbf{E}_t\Sigmac^{-1}\mathbf{E}_t^\top).
\]
Stacking the four blocks and summing over $t = 2,\ldots,T$ recovers equation~\eqref{eq:score} in the main text.

\section{Proof of Proposition \ref{prop:jacobian}}
\label{sec:proof}

\begin{proof}
  The proof proceeds by applying the chain rule through four steps:
  
  $\hat{\bm{\Sigma}}_{\text{c}} \to \hat{\bm{\Sigma}}_{\text{c}}^{-1} \to \hat{\mathbf{C}}_0^\top \to \hat{\mathbf{Q}}^{-\top} \to
  \operatorname{vecl}(\hat{\mathbf{Q}}^{-\top})$.

  The first step concerns the Jacobian of matrix inversion,
  $\frac{\partial\,\operatorname{vech}(\hat{\bm{\Sigma}}_{\text{c}}^{-1})}{\partial\,\operatorname{vech}(\hat{\bm{\Sigma}}_{\text{c}})^\top}$. From
  the differential
  $d(\hat{\bm{\Sigma}}_{\text{c}}^{-1}) = -\hat{\bm{\Sigma}}_{\text{c}}^{-1}(d\hat{\bm{\Sigma}}_{\text{c}})\hat{\bm{\Sigma}}_{\text{c}}^{-1}$,
  vectorizing and reducing to $\operatorname{vech}$ via the duplication matrix $\mathbf{D}_n$ and its Moore--Penrose inverse $\mathbf{D}_n^+$:

  \[
    \mathbf{J}_1 =
    \frac{\partial\,\operatorname{vech}(\hat{\bm{\Sigma}}_{\text{c}}^{-1})}{\partial\,\operatorname{vech}(\hat{\bm{\Sigma}}_{\text{c}})^\top} =
    -\mathbf{D}_n^+\,(\hat{\bm{\Sigma}}_{\text{c}}^{-1}\otimes\hat{\bm{\Sigma}}_{\text{c}}^{-1})\,\mathbf{D}_n.
  \] 
  Then, we compute the Jacobian of the Cholesky factor,
  $\frac{\partial\,\operatorname{vec}(\hat{\mathbf{C}}_0^\top)}{\partial\,\operatorname{vech}(\hat{\bm{\Sigma}}_{\text{c}}^{-1})^\top}$. From the
  Cholesky decomposition of the precision matrix $\hat{\bm{\Sigma}}_{\text{c}}^{-1} = \hat{\mathbf{C}}_0^\top\hat{\mathbf{C}}_0$, an infinitesimal perturbation
  gives
  \[
    d\hat{\bm{\Sigma}}_{\text{c}}^{-1} = d\hat{\mathbf{C}}_0^\top\,\hat{\mathbf{C}}_0 + \hat{\mathbf{C}}_0^\top\,d\hat{\mathbf{C}}_0.
  \]
  Vectorizing both sides and using $\operatorname{vec}(\mathbf{ABC}) = (\mathbf{C}^\top\otimes \mathbf{A})\operatorname{vec}(\mathbf{B})$ and
  $\operatorname{vec}(\mathbf{A}^\top)=\mathbf{K}_{nn}\operatorname{vec}(\mathbf{A})$ yields
  \[
    \operatorname{vec}(d\,\hat{\bm{\Sigma}}_{\text{c}}^{-1}) = \left[(\hat{\mathbf{C}}_0^\top\otimes\mathbf{I}_n) +
      (\mathbf{I}_n\otimes\hat{\mathbf{C}}_0^\top)\mathbf{K}_{nn}\right]\operatorname{vec}(d\hat{\mathbf{C}}_0^\top).
  \]
  
  Since $\hat{\mathbf{C}}_0^\top$ is lower triangular, we use the elimination matrix to write
  $\operatorname{vec}(d\hat{\mathbf{C}}_0^\top) = \mathbf{T}_n\,\operatorname{vech}(d\hat{\mathbf{C}}_0^\top)$. Applying $\mathbf{D}_n^+$ to reduce
  the symmetric $\operatorname{vec}(d\,\hat{\bm{\Sigma}}_{\text{c}}^{-1})$ to $\operatorname{vech}(d\,\hat{\bm{\Sigma}}_{\text{c}}^{-1})$, we obtain
  \[
    \operatorname{vech}(d\,\hat{\bm{\Sigma}}_{\text{c}}^{-1}) = \hat{\mathbf{M}}\,\operatorname{vech}(d\hat{\mathbf{C}}_0^\top),
  \]
  where
  $\hat{\mathbf{M}} = \mathbf{D}_n^+\left[(\hat{\mathbf{C}}_0^\top\otimes\mathbf{I}_n) +
    (\mathbf{I}_n\otimes\hat{\mathbf{C}}_0^\top)\mathbf{K}_{nn}\right]\mathbf{T}_n$ is a square invertible matrix of dimension
  $\frac{n(n+1)}{2}\times\frac{n(n+1)}{2}$. Solving for $\operatorname{vec}(d\hat{\mathbf{C}}_0^\top)$ and noting that
  $\operatorname{vech}(d\,\hat{\bm{\Sigma}}_{\text{c}}^{-1}) = d\operatorname{vech}(\hat{\bm{\Sigma}}_{\text{c}}^{-1})$, we get
  \[
    \operatorname{vec}(d\hat{\mathbf{C}}_0^\top) = \mathbf{T}_n\,\hat{\mathbf{M}}^{-1}\,d\operatorname{vech}(\hat{\bm{\Sigma}}_{\text{c}}^{-1}), 
  \]
  which implies
  \[
    \mathbf{J}_2 =
    \frac{\partial\,\operatorname{vec}(\hat{\mathbf{C}}_0^\top)}{\partial\,\operatorname{vech}(\hat{\bm{\Sigma}}_{\text{c}}^{-1})^\top} =
    \mathbf{T}_n\,\hat{\mathbf{M}}^{-1}.
  \] 
  
  The third step foresees the computation of the Jacobian of normalized inversion,
  $\frac{\partial\,\operatorname{vec}(\hat{\mathbf{Q}}^{-\top})}{\partial\,\operatorname{vec}(\hat{\mathbf{C}}_0^\top)^\top}$. From Equation
  \eqref{eq:Q}, the unit lower triangular structural impact matrix is given by
  $\hat{\mathbf{Q}}^{-\top} = \hat{\mathbf{C}}_0^{-\top}\langle\hat{\mathbf{C}}_0\rangle$. Differentiating this mapping yields:
  \[
    d\hat{\mathbf{Q}}^{-\top} = (d\hat{\mathbf{C}}_0^{-\top})\langle\hat{\mathbf{C}}_0\rangle +
    \hat{\mathbf{C}}_0^{-\top}(d\langle\hat{\mathbf{C}}_0\rangle). 
  \]
  Substituting $d\hat{\mathbf{C}}_0^{-\top} = -\hat{\mathbf{C}}_0^{-\top}(d\hat{\mathbf{C}}_0^{\top})\hat{\mathbf{C}}_0^{-\top}$, we get
  \[
    d\hat{\mathbf{Q}}^{-\top} = -\hat{\mathbf{C}}_0^{-\top}(d\hat{\mathbf{C}}_0^{\top})\hat{\mathbf{C}}_0^{-\top}\langle\hat{\mathbf{C}}_0\rangle +
    \hat{\mathbf{C}}_0^{-\top}(d\langle\hat{\mathbf{C}}_0\rangle). 
  \]
  Since $\hat{\mathbf{Q}}^{-\top} = \hat{\mathbf{C}}_0^{-\top}\langle\hat{\mathbf{C}}_0\rangle$, we have
  \[
    d\hat{\mathbf{Q}}^{-\top} = -\hat{\mathbf{C}}_0^{-\top}(d\hat{\mathbf{C}}_0^{\top})\hat{\mathbf{Q}}^{-\top} +
    \hat{\mathbf{C}}_0^{-\top}(d\langle\hat{\mathbf{C}}_0\rangle).
  \]
Vectorizing and using $\operatorname{vec}(\mathbf{ABC}) = (\mathbf{C}^{\top} \otimes \mathbf{A})\operatorname{vec}(\mathbf{B})$ yields
\[
  \operatorname{vec}(d \hat{\mathbf{Q}}^{-\top}) = -\left(\hat{\mathbf{Q}}^{-1} \otimes \hat{\mathbf{C}}_0^{-\top}\right)
  \operatorname{vec}(d\hat{\mathbf{C}}_0^\top) + \left(\mathbf{I}_n \otimes \hat{\mathbf{C}}_0^{-\top}\right)\operatorname{vec}(d
  \langle\hat{\mathbf{C}}_0\rangle).
\]
Since $\langle\hat{\mathbf{C}}_0\rangle$ is a diagonal matrix, we can select its diagonal elements from $\hat{\mathbf{C}}_0^\top$ using selection
matrices $\mathbf{H}_k=\bm{\eta}_k\bm{\eta}_k^\top$, where $\bm{\eta}_k$ is the $k$-th canonical basis vector of $\mathbb{R}^n$. Specifically,
$\operatorname{vec}(d
\langle\hat{\mathbf{C}}_0\rangle)=\sum_{k=1}^n(\mathbf{H}_k\otimes\mathbf{H}_k)\operatorname{vec}(d\hat{\mathbf{C}}_0^\top)$. Substituting this
into the second term yields
\[
\begin{split}
  &\left(\mathbf{I}_n \otimes \hat{\mathbf{C}}_0^{-\top}\right)\sum_{k=1}^n\left(\mathbf{H}_k \otimes \mathbf{H}_k\right) \operatorname{vec}(d \hat{\mathbf{C}}_0^{\top}) =\\
  &\sum_{k=1}^n\left(\mathbf{H}_k \otimes \hat{\mathbf{C}}_0^{-\top}\mathbf{H}_k\right) \operatorname{vec}(d \hat{\mathbf{C}}_0^{\top}).
\end{split}
\]
Therefore, combining both terms and factoring out $\operatorname{vec}(d \hat{\mathbf{C}}_0^\top)$ we obtain:
\[
    \mathbf{J}_3 = \sum_{k=1}^n\left(\mathbf{H}_k \otimes \hat{\mathbf{C}}_0^{-\top}\mathbf{H}_k\right)  - \left(\hat{\mathbf{Q}}^{-1} \otimes \hat{\mathbf{C}}_0^{-\top}\right).
\]

Combining via the chain rule:
\[
\begin{split}
  \mathbf{J} &= \mathbf{L}_n^s \cdot \mathbf{J}_3 \cdot \mathbf{J}_2 \cdot \mathbf{J}_1\\
  &= \mathbf{L}_n^s \left[ \sum_{k=1}^n\bigl(\mathbf{H}_k \otimes \hat{\mathbf{C}}_0^{-\top}\mathbf{H}_k\bigr) - \bigl(\hat{\mathbf{Q}}^{-1} \otimes \hat{\mathbf{C}}_0^{-\top}\bigr) \right]
  \left(\mathbf{T}_n\,\hat{\mathbf{M}}^{-1}\right)\left(-\mathbf{D}_n^+\,(\hat{\bm{\Sigma}}_{\text{c}}^{-1}\otimes\hat{\bm{\Sigma}}_{\text{c}}^{-1})\,\mathbf{D}_n\right)\\
  &= \mathbf{L}_n^s \left[ \bigl(\hat{\mathbf{Q}}^{-1} \otimes \hat{\mathbf{C}}_0^{-\top}\bigr) -\sum_{k=1}^n\bigl(\mathbf{H}_k \otimes \hat{\mathbf{C}}_0^{-\top}\mathbf{H}_k\bigr)\right]
  \mathbf{T}_n\,\hat{\mathbf{M}}^{-1}\mathbf{D}_n^+\,(\hat{\bm{\Sigma}}_{\text{c}}^{-1}\otimes\hat{\bm{\Sigma}}_{\text{c}}^{-1})\,\mathbf{D}_n,
\end{split}
\] 
which exactly matches the result stated in Proposition \ref{prop:jacobian}.
\end{proof}
\end{document}